%% file: TCNS.tex
\documentclass[10pt,twocolumn,twoside]{IEEEtran}
\usepackage{cite}
\usepackage{graphicx}
\usepackage{subcaption}
\usepackage[cmex10]{amsmath}
\usepackage{algpseudocode}
\usepackage{algorithmicx} 
\usepackage{array}
\usepackage{url}
\usepackage{caption}
\usepackage{epsfig}
\usepackage{amsfonts,amssymb,multirow,bigstrut,booktabs,ctable,latexsym}

\usepackage{cuted}
\usepackage{flushend}

\usepackage[mathscr]{eucal}
\usepackage{verbatim}
\usepackage{amsthm}
\usepackage{algorithm}

\begin{document}

\newtheorem{definition}{Definition}
\newtheorem{lemma}{Lemma}
\newtheorem{theorem}{Theorem}
\newtheorem{example}{Example}
\newtheorem{proposition}{Proposition}
\newtheorem{remark}{Remark}
\newtheorem{assumption}{Assumption}
\newtheorem{corrolary}{Corrolary}
\newtheorem{property}{Property}
\newtheorem{ex}{EX}
\newtheorem{problem}{Problem}
\newcommand{\argmin}{\arg\!\min}
\newcommand{\argmax}{\arg\!\max}
\newcommand{\st}{\text{s.t.}}
\newcommand \dd[1]  { \,\textrm d{#1}  }

\newcounter{mytempeqncnt}

\title{Swarm-STL: A Framework for Motion Planning\\ in Large-Scale, Multi-Swarm Systems}  

\author{Shiyu Cheng, \IEEEmembership{Graduate Student Member, IEEE}, Luyao Niu, \IEEEmembership{Member, IEEE}, Bhaskar Ramasubramanian, \IEEEmembership{Member, IEEE}, Andrew Clark, \IEEEmembership{Senior Member, IEEE}, and Radha Poovendran, \IEEEmembership{Fellow, IEEE}
\thanks{S. Cheng and A. Clark are with the Department of Electrical and Systems Engineering, Washington University in St. Louis, St. Louis, MO, USA. Email: \{cheng.shiyu, andrewclark\}@wustl.edu}
\thanks{L. Niu is with the Department of Electrical and Computer Engineering, University of Washington, Seattle, WA, USA. Email: luyaoniu@uw.edu}
\thanks{B. Ramasubramanian is with the Electrical and Computer Engineering Program, Western Washington University,
Bellingham, WA, USA. Email: ramasub@wwu.edu}
\thanks{R. Poovendran is with the Network Security Lab, Department of Electrical and Computer Engineering, University of Washington, Seattle, WA, USA. Email: rp3@uw.edu}
}

\maketitle
\begin{abstract}
In multi-agent systems, signal temporal logic (STL) is widely used for path planning to accomplish complex objectives with formal safety guarantees. However, as the number of agents increases, existing approaches encounter significant computational challenges. Recognizing that many complex tasks require cooperation among multiple agents, we propose swarm STL specifications to describe the collective tasks that need to be achieved by a team of agents.
Next, we address the motion planning problem for all the agents in two stages. First, we abstract a group of cooperating agents as a swarm and construct a reduced-dimension state space whose dimension does not increase with the number of agents. 
The path planning is performed at the swarm level, ensuring the safety and swarm STL specifications are satisfied. Then, we design low-level control strategies for agents within each swarm based on the path synthesized in the first step. The trajectories of agents generated by the two-step policy ensure satisfaction of the STL specifications. We evaluate our two-stage approach in both single-swarm and multi-swarm scenarios. The results demonstrate that all tasks are completed with safety guarantees. Compared to the baseline multi-agent planning approach, our method maintains computational efficiency as the number of agents increases, since the computational time scales with the number of swarms rather than the number of agents.
\end{abstract}


\input{Introduction}
\input{Related}

\input{Preliminaries}
\input{System}
\input{Solution}
\input{Simulation}
\input{Conclusion}

\IEEEpeerreviewmaketitle

\bibliographystyle{IEEEtran}
\bibliography{MyBib}

\input{appendix}
\end{document}

%% file: Introduction.tex
\section{Introduction}

Multi-agent autonomous systems, such as search and rescue robots, are often required to cooperate to accomplish complex and time-sensitive tasks. 
For example, a specification may require multiple agents to simultaneously occupy a set of desired regions in the environment.
Such specifications can be formalized using signal temporal logic (STL) \cite{maler2004monitoring}.
A naive approach to synthesize control inputs to ensure autonomous systems satisfy the specifications is to treat the multi-agent systems as a ``meta-agent" whose state is the concatenation of all agents' states, and synthesize the control input for the meta-agent.
Existing solutions include discretization-based \cite{kress2018synthesis, fainekos2009temporal}, barrier function-based \cite{lindemann2018control}, and model predictive control (MPC)-based \cite{raman2014model} approaches.

However, as the number of agents increases, the dimension of the meta-agent states increases significantly, making existing correct-by-construction control synthesis \cite{buyukkocak2021planning} to satisfy temporal logic specifications intractable.
To address the scalability challenge, several approaches have been developed \cite{chen2020guaranteed, fu2020distributed, lindemann2020barrier}. 
However, these approaches assume that there is a central entity to decompose the temporal logic specification into a set of tasks, and assign the tasks to agents.
Due to dependencies among the tasks and couplings between the agents, decomposing specification and assigning tasks in large-scale multi-agent systems operating in complex environments are often computationally expensive and may not always be feasible.
Consequently, the aforementioned approaches may not be applicable to large-scale multi-agent systems.

In this paper, we consider a large-scale multi-agent system subject to an STL specification. 
The STL specification requires the agents to satisfy a set of desired properties without imposing constraints based on identities of agents. Our objective is to synthesize control inputs for the agents to meet the STL specification. Considering that each task requires a specified number of agents to be completed, we first propose swarm STL specifications to describe the objectives that all agents within a swarm are required to achieve.
Our key insight to overcome the scalability challenge is to synthesize control inputs on a dimension-reduced state space, where we collectively control swarms of agents. 
We characterize the swarm behaviors by deriving a minimum volume ellipsoid bounding all agents in a reduced-dimension state space. 
Given the ellipsoids, the problem of satisfying the STL specification is converted to the problem of collectively controlling agents by configuring the centroids and shapes of bounding ellipsoids. Hence, the dimensionality of the problem depends on the number of swarms rather than the number of agents.

We develop an iterative algorithm to generate a sequence of waypoints with timestamps for the ellipsoids.
We show that our synthesized waypoints can always be reached by the swarms.
Furthermore, if the swarms can reach the synthesized waypoints within desired time intervals, then the original STL specification is satisfied.
We then derive a decentralized control input for each agent within the swarm.
We demonstrate our proposed framework using multiple state-of-the-art benchmarks.
To summarize, this paper makes the following contributions:
\begin{itemize}
    \item We develop a provably correct procedure to convert an STL specification to an equivalent swarm STL.
    \item We formulate an optimization program to synthesize the waypoints, timestamps, and bounding ellipsoids for all swarms. We formally verify that the synthesized waypoints always result in feasible control inputs for all agents, ensuring satisfaction of the swarm temporal logic specification.
    \item We synthesize a collective control for each swarm to reach the desired waypoints. We compute the control input for each agent based on the collective control.
    \item We validate our proposed approach across several benchmarks, demonstrating its effectiveness. 
   We also demonstrate that our approach is significantly more efficient than the baseline in the multi-agent scenario when the number of agents is at least six.
\end{itemize}

The remainder of this paper is organized as follows.
Section \ref{sec:related} reviews related literature.
Section \ref{sec:prelim} presents preliminaries on STL specifications.
We discuss the system model and present problem formulation in Section \ref{sec:formulation}. Section \ref{sec:solution} presents our proposed methods for swarm-level planning and low-level control at the agent level. Case studies are presented in Section \ref{sec:simulation}. In Section \ref{sec: conclusion}, we conclude the paper.

%% file: Related.tex
\section{Related Work}\label{sec:related}

Control synthesis under temporal logic constraints has been extensively studied \cite{kress2018synthesis, tellex2020robots, fan2020fast, raman2014model, lindemann2018control, fainekos2009temporal}.
The authors of \cite{kress2018synthesis, fainekos2009temporal} employ discretization-based approaches to partition the state space and synthesize control inputs to meet temporal logic constraints.
Model predictive control (MPC)-based approaches \cite{raman2014model, farahani2015robust,sadraddini2015robust, wongpiromsarn2012receding} have also been widely used to satisfy temporal logic constraints. 
These approaches discretizes the time horizon and computes a control input at each discrete time interval. 
Instead of solving a large optimization program as per MPC-based approaches, barrier function-based approaches \cite{lindemann2018control,zehfroosh2022non} formulate a quadratic program at each time interval to synthesize control inputs. 

Discretization-based \cite{buyukkocak2021planning, kloetzer2011multi}, MPC-based \cite{sahin2019multirobot,jones2019scalable}, and barrier function-based \cite{chen2020guaranteed, fu2020distributed, lindemann2020barrier} approaches have been extended to synthesize control inputs for multi-agent systems under temporal logic constraints. 
Different from the existing work, we do not require a central entity to decompose the STL specification and assign tasks to agents based on their identities. 
Instead, we focus on swarms in a dimension-reduced state space and develop a scalable approach to automatically synthesize control inputs for all agents. A centroid with a bounding ellipsoid is used for path planning at the swarm level. A related concept appears in stochastic control, where the probability distribution of the system state is analyzed and manipulated \cite{chen2015optimal}\cite{zhang2025distributional}.
In addition, unlike the MPC-based and barrier function-based approaches, which may not always guarantee the existence of feasible control inputs to satisfy the STL specification, we formally prove the feasibility of control inputs to satisfy the STL specification.
\input{tables/notation_table}

%% file: tables/notation_table.tex
\begin{table}[h]
    \centering
    \begin{tabular}{|c|l|}
        \hline
        \textbf{Symbol} & \textbf{Description} \\
        \hline
        $x_i$ & State of agent $i$ \\
        $p_i$ & Position of agent $i$\\
        $v_i$ & Velocity of agent $i$ \\
        $u_i$ & Control input of agent $i$ \\
        $A, B$ & System matrices defining agent dynamics \\
        $\bar{x}_s$ & Centroid of swarm $s$ \\
        $\psi$ & Swarm signal temporal logic (STL) specification \\
        $\pi$ & Safety specification \\
        $\mu$ & Atomic predicate in STL \\
        $\mathcal{S}$ & Set of swarms \\
        $\mathcal{N}$ & Set of all agents \\
         $\underline{\mathcal{S}}(N_\mu)$ & Set of swarms with fewer than $N_\mu$ agents\\
          $\overline{\mathcal{S}}(N_\mu)$& Set of swarms with at least $N_\mu$ agents\\
        $\Sigma_s$ & Bounding ellipsoid matrix for swarm $s$ \\
        $t_{s,k}$ & Timestamp for segment $k$ of swarm $s$ \\
        $\eta$ & Maximum tracking error of centroids\\
        $\zeta$ & Minimal safety distance between agents\\
        $\epsilon$ & Feasible margin \\
        $\chi$ & Actuation-induced motion bound\\
        $\xi$ & Volume inflation factor for bounding ellipsoid\\
        $\Gamma_{s,k}$ & $\eta$-tube around the reference path for segment $k$ of swarm $s$  \\
        \hline
    \end{tabular}
    \caption{Notation Table}
\end{table}

%% file: Preliminaries.tex
\section{Preliminaries}\label{sec:prelim}

This section presents preliminaries on STL specifications. An STL specification \cite{maler2004monitoring} is defined over a signal $x(t): \mathbb{R}_{\geq 0}\rightarrow\mathbb{R}^n$  as follows:
\begin{equation*}
    \phi = \text{True} | \mu | \neg\phi | \phi_1\land\phi_2 | \phi_1 U_{[a,b]}\phi_2,
\end{equation*}
where $\phi_1$ and $\phi_2$ are STL specifications.
In the above, $\mu$ is an atomic predicate that can be evaluated based on the value of a predicate function $h^{\mu}:\mathbb{R}^n\rightarrow \mathbb{R}$. The definition of $h^{\mu}$ is as follows:
\begin{equation*}
    \mu = \begin{cases}
        \text{True}&\mbox{ if }h^{\mu}(x)\geq 0\\
        \text{False}&\mbox{ if }h^{\mu}(x)< 0
    \end{cases}
    \end{equation*}
    We define $U_{[a,b]}$ as the until operator, where $a,b\in\mathbb{R}_{\geq 0}$ with $a\leq b$.
An STL specification can be expressed in negation normal form (NNF), where negations are pushed inward to only appear directly in front of atomic predicates.
We indicate that a signal $x(t)$ satisfies an STL specification $\phi$ at time $t$ as $(x,t)\models\phi$. 
Other Boolean operators such as disjunction $\lor$ and temporal operators including eventually ($F_{[a,b]}$) and always ($G_{[a,b]}$) can be defined inductively.
The semantics of STL \cite{lindemann2018control} are given below.

\begin{definition}[Semantics of STL]
    Consider a signal $x: \mathbb{R}_{\geq 0}\rightarrow \mathbb{R}^n$, the semantics of STL are defined as 
    \begin{align*}
        (x, t)&\models\mu && \Leftrightarrow h^{\mu}(x(t))\geq 0\\
        (x, t)&\models\neg \phi && \Leftrightarrow \neg((x,t)\models\phi)\\
        (x, t)&\models \phi_1\wedge \phi_2 &&\Leftrightarrow (x, t)\models \phi_1\wedge (x, t)\models\phi_2\\
        (x, t)&\models \phi_1U_{[a,b]}\phi_2 &&\Leftrightarrow \exists t_{1} \in [t+a,t+b] \text{s.t. } (x,t_{1}) \models \phi_2\\
        &\qquad \qquad&&\quad\quad\wedge \forall t_2\in [t,t_1], (x,t_2)\models \phi_1\\
        (x, t)&\models F_{[a,b]}\phi && \Leftrightarrow \exists t_1\in [t+a, t+b] \  \text{s.t. } (x,t_1)\models \phi\\
        (x,t)&\models G_{[a,b]}\phi && \Leftrightarrow \forall t_1\in [t+a, t+b], (x,t_1)\models \phi
    \end{align*}
\end{definition}

%% file: System.tex
\section{System Model and Problem Formulation}\label{sec:formulation}

In this section, we first model the multi-agent system. 
Next, we describe the STL specification that the agents need to satisfy.
We finally present the problem formulation.

\subsection{System Model}
We consider a collection of agents, denoted as $\mathcal{N}=\{1,\ldots,N\}$.
Each agent $i$ follows double integrator dynamics 
\begin{equation}
\label{eq:agent-dynamics}
    \dot{x}_i =\begin{bmatrix}
        \dot{p}_i\\ \dot{v}_i
    \end{bmatrix}= Ax_i + Bu_i,
\end{equation}
where $x_i\in\mathcal{X}\subset\mathbb{R}^{2d}$ is the state of agent $i$, $p_i\in \tilde{\mathcal{P}}\subset \mathbb{R}^d$ is the position, $v_i\in \mathcal{V}\subset \mathbb{R}^d$ is the velocity,  $u_i\in\mathcal{U}_i\subset\mathbb{R}^d$ is the control input, $d$ denotes the dimension of the position and velocity, and $A = \begin{bmatrix}
    0 & I\\ 0 & 0
\end{bmatrix}\in\mathbb{R}^{2d\times 2d}$, $B= \begin{bmatrix}
    0\\I
\end{bmatrix}\in\mathbb{R}^{2d\times d}$. 
Given an initial system state $x_i$ and the control input $u_i$, we denote the state of agent $i$ at any time $t\geq 0$ as $x_i(t)$.
We let $\mathbf{x}(t) = [x_1^T(t),\ldots, x_N^T(t)]^T$ for any time $t\geq 0$. 

The agents  form a set of swarms, denoted by $\mathcal{S}=\{1,\ldots,S\}$.
We denote the set of agents belonging to a swarm $s\in\mathcal{S}$ by $\mathcal{N}_s$.
We assume that each agent must belong to exactly one swarm.
That is, $\cup_{s\in\mathcal{S}}\mathcal{N}_s=\mathcal{N}$ and $\mathcal{N}_s\cap\mathcal{N}_{s'}=\emptyset$ for all $s,s'\in\mathcal{S}$ with $s \neq s^{\prime}$.

\subsection{Swarm STL Specifications}

In what follows, we first define safety specifications and swarm STL specifications.
Then, we describe the STL specifications given to the multi-agent system.
We finally present the problem formulation. 
\begin{definition}
\label{def:safe-STL}
    An STL specification $\pi$ is a safety specification if it can be written in the form $G_{[a,b]}\pi_{1} \wedge G_{[a,b]}\pi_{2}$, where 
    \begin{displaymath}
        \pi_{1} = \bigwedge_{1 \leq i < j \leq N}{\mu_{ij}}, \quad \pi_{2} = \bigwedge_{i=1}^{N}{\bigwedge_{j=1}^{M}{\theta_{ij}}}.
    \end{displaymath}
    In the above, $\{\mu_{ij}: i,j=1,\ldots,N, i< j\}$ is a set of atomic predicates with predicate functions 
    $h_{ij}^{\mu}(\mathbf{x}(t)) = \|p_{i}(t)-p_{j}(t)\|_{2} - \zeta$
    for the given $\zeta > 0$ and $\{\theta_{ij}: i=1,\ldots,N, j=1,\ldots,M\}$ is a set of atomic predicates with predicate functions $h_{ij}^{\theta}(\mathbf{x}(t)) = \gamma_{j}(x_{i}(t))$ for some continuous functions $\gamma_{1},\ldots,\gamma_{M}$.
\end{definition}
In Definition \ref{def:safe-STL}, $\pi_{1}$ defines a lower bound $\zeta$ on the distances between each pair of agents, while $\pi_{2}$ defines safety constraints (e.g., obstacle avoidance) that must be satisfied by each agent at each time step. 
In this paper, we consider functions $\gamma_{1},\ldots,\gamma_{M}$ that are of the linear form $\gamma_j(x_{i}(t)) = a_j^Tp_{i}(t) + b_j$, where $ j\in \{1, \ldots, M\}$ for $h_{ij}^{\theta}(\mathbf{x}(t))$. We adopt linear function representations to construct obstacles as convex polytopes.

We next define \emph{swarm STL specifications} to describe objectives that can be achieved by teams of agents. 
We extend the definition of predicates in STL as follows.
\begin{definition}\label{def:swarm predicate}
    A predicate $\mu$ is a swarm predicate if there exists a function $h^{\mu}: \mathbb{R}^{2d} \rightarrow \mathbb{R}$ and an integer $N_{\mu} \in \{1,\ldots,N\}$ such that 
    \begin{equation*}
    \mu = 
    \begin{cases}
    \emph{True}, & \mbox{\emph{if} } |\{i: h^{\mu}(x_{i}(t)) \geq 0\}| \geq N_{\mu} \\
    \emph{False}, & \emph{\text{else}}
    \end{cases}
    \end{equation*}
\end{definition}
Intuitively, a predicate is a swarm predicate if it holds when a particular number of agents $N_{\mu}$ satisfies a desired state constraint, without imposing additional limitations on \emph{which} agents satisfy the constraint. 
We consider linear predicate functions $h^{\mu}(x_{i}(t)) = F^Tp_{i}(t)+g$, where $F\in\mathbb{R}^{d}$ and $g\in \mathbb{R}$. A swarm STL specification is defined as $\psi$, in which each predicate is a swarm predicate given in Definition \ref{def:swarm predicate}.

Given the definitions of the safety and swarm STL specifications, the specification for the multi-agent system can be expressed as a conjunction between them, as follows:
\begin{equation}\label{eq:STL}
    \phi = \psi \land \pi,
\end{equation}
where $\psi$ is a swarm STL specification, and $\pi$ is a safety STL specification.
We denote the time horizon considered by $\phi$ as $[t_0,t_0+T]$.

We summarize the problem studied in this paper as follows.
\begin{problem}\label{prob:formulation}
    Consider a set of agents $\mathcal{N}$ who are given a mission modeled as Eq. \eqref{eq:STL}. Compute a control input $\{u_{i}(t): t \in [t_{0},t_{0}+T]\}$ for $i = 1, \ldots, N$, such that $(\mathbf{x},t) \models \phi$.
\end{problem}

%% file: Solution.tex
\section{Proposed Solution}\label{sec:solution}

In this section, we describe our solution to Problem \ref{prob:formulation}.
We first convert specification $\phi$ to an equivalent expression in a reduced-dimension state space, describing the desired behaviors of swarms of agents.
We show that it suffices to satisfy the equivalent specification in the reduced-dimension state space to meet $\phi$. We develop algorithms for trajectory planning in the reduced-dimension space, and then propose low-level control laws for the agents to implement the planned trajectories. 
To simplify notation, we omit the explicit time dependence throughout the rest of the paper, unless otherwise stated.

\subsection{Reduced-Dimension State Space Construction}

We represent each swarm $s\in\mathcal{S}$ using a time-varying bounding ellipsoid enclosing all agents within the swarm. 
The centroid of the bounding ellipsoid for swarm $s$ is defined as 
\begin{equation}
    \overline{x}_{s} \triangleq \frac{1}{|\mathcal{N}_{s}|}\sum_{i\in \mathcal{N}_{s}}x_{i}.    
\end{equation}
We can describe the dynamics of $\overline{x}_{s}$ as 
\begin{equation}\label{eq:dynamics-centroid}
    \dot{\overline{x}}_{s} = \frac{1}{|\mathcal{N}_{s}|}\sum_{i\in \mathcal{N}_{s}}\dot{x}_{i}= A\overline{x}_{s} + B\overline{u}_{s},
\end{equation}
where $\overline{u}_{s} = \frac{1}{|\mathcal{N}_{s}|}\sum_{i\in \mathcal{N}_{s}}u_{i}$ is the collective control input of all agents within the swarm. The bounding ellipsoid is defined as 
\begin{multline}\label{eq:bounding ellipsoid}
    \overline{\Omega}(\overline{x}_{s},\Sigma_{s})\triangleq\{\begin{bmatrix}
        p\\v
    \end{bmatrix}\in\mathcal{X} :(p-\overline{p}_s)^T\Sigma_{s}^{-1}(p-\overline{p}_s)\leq 1,\\ p\in\mathbb{R}^d, v\in\mathbb{R}^d\},
\end{multline}
where $\Sigma_{s}\in\mathbb{S}_{++}^d$ is a positive definite matrix.

Using the centroid position $\bar{p}_s$ and bounding ellipsoid matrix $\Sigma_{s}$, we can characterize the behavior exhibited by each swarm $s$ in a reduced-dimension state space $(\bar{x}_s,\Sigma_{s})\in\mathbb{R}^{2d}\times \mathbb{S}_{++}^d$.
We remark that the dimensionality of this state space does not increase as the number of agents increases.
In the remainder of this section, we convert the specification $\phi$ in Eq. \eqref{eq:STL} to an equivalent expression that can be efficiently evaluated in the reduced-dimension state space.

\subsection{STL Specification in the Reduced-Dimension State Space}\label{sec:reduced STL}

We denote the set of swarms with fewer than $N_\mu$ and at least $N_\mu$ agents as $\underline{\mathcal{S}}(N_\mu)$ and $\overline{\mathcal{S}}(N_\mu)$, respectively. 
In what follows, we express the STL specification $\phi$ in the reduced-dimension state space. 

Without loss of generality, we assume that $\psi$ in Eqn. \eqref{eq:STL} is in NNF. For each predicate $\mu$ involved in $\psi$, 
we define a set of new predicates $\overline{\mu}^{1},\ldots,\overline{\mu}^{S}$ by functions $\overline{h}_{s}^{\mu} : \mathbb{R}^{2d} \times \mathbb{S}_{++}^{d} \rightarrow \mathbb{R}$, which can be evaluated using
\begin{multline}\label{eq:h mu}
    \overline{h}_s^{\mu}(\overline{x}_s,\Sigma_s) = \min\left\{h^{\mu}(\begin{bmatrix}
        p\\ v
    \end{bmatrix}) : \right. \\
    \left.(p-\overline{p}_{s})^{T}\Sigma_{s}^{-1}(p-\overline{p}_{s}) \leq 1,p\in \mathbb{R}^d, v\in \mathbb{R}^{d}\right\}.
\end{multline}
That is, if $\overline{h}_s^{\mu}(\overline{x}_s,\Sigma_s)$ is non-negative, then all agents within swarm $s$ satisfy  $h^\mu(x_i)\geq 0$.
Therefore, when swarm $s$ contains at least $N_\mu$ agents, swarm predicate $\mu$ presented in Definition \ref{def:swarm predicate} is true.
We can now transform $\psi$ into a new specification $\overline{\psi}$ by replacing each atomic predicate $\mu$ in $\psi$ according to the following rules:
\begin{itemize}
    \item If a predicate $\mu$ appears without negation, that predicate is replaced by 
    \begin{equation*}
        \bigvee_{s \in \overline{\mathcal{S}}(N^{\mu})}{\overline{\mu}^{s}}.
    \end{equation*}
    \item If a negated predicate $\neg\mu$ appears, the negated predicate is replaced by $$\left(\bigwedge_{s}{\neg\overline{\mu}^{s'}}\right)\lor\left(\bigvee_{s \in \underline{\mathcal{S}}(N^{\mu})}{\left(\overline{\mu}^{s} \land \left(\bigwedge_{s' \neq s}{\neg\overline{\mu}^{s'}}\right)\right)}\right).$$
\end{itemize}

Next, we convert the safety specification $\pi$ to one defined over the reduced-dimension state space. 
We rewrite $\pi$ as 
\begin{equation*}
    G_{[t_0,t_0+T]}\overline{\pi}_{1} \land G_{[t_0,t_0+T]}\overline{\pi}_{2},
\end{equation*}
where 
\begin{equation*}
    \overline{\pi}_{1} = \bigwedge_{s,s'\in\mathcal{S},s\neq s'}{\overline{\mu}_{ss'}}, \quad \overline{\pi}_{2} = \bigwedge_{s=1}^{S}{\bigwedge_{j=1}^{M}{\overline{\theta}_{sj}}}.
\end{equation*}
Here, the new predicate $\overline{\mu}_{ss'}$ is evaluated using predicate function
\begin{multline*}
\overline{h}_{ss'}^{\mu}(\bar{x}_s,\bar{x}_{s^{\prime}})= \min\{\|p_{i}-p_{j}\|_{2} - \zeta: \\(p_{i}-\overline{p}_{s})^{T}\Sigma_{s}^{-1}(p_{i}-\overline{p}_{s}) \leq 1, 
(p_{j}-\overline{p}_{s'})^{T}\Sigma_{s'}^{-1}(p_{j}-\overline{p}_{s'}) \leq 1\},
\end{multline*}
while $\overline{\theta}_{sj}$ has function
\begin{multline}
    \overline{h}_{sj}^{\theta}(\overline{x}_{s},\Sigma_{s}) = \min\left\{ \gamma_{j}(\begin{bmatrix}
        p\\
        v
    \end{bmatrix}) : \right. \\
    \left. (p - \overline{p}_{s})^{T}\Sigma_{s}^{-1}(p - \overline{p}_{s}) \leq 1, p \in \mathbb{R}^d,\ v \in \mathbb{R}^{d}\right\}.
\end{multline}
Intuitively, if $\overline{h}_{ss'}^{\mu}\geq 0$ for all swarms $s\neq s'$, then the distance between all pairs of agents is lower bounded by $\zeta$. Thus, $\pi_1$ in Definition \ref{def:safe-STL} is evaluated as true for any agent $i$ and $j$ belonging to different swarms.
If function $\overline{h}_{sj}^{\theta}(\overline{x}_{s},\Sigma_{s})\geq 0$ for all swarms $s\in\mathcal{S}$, it indicates that all agents $i$ within swarm $s$ ensures $\gamma_j(x_i)\geq 0$, and hence $\pi_2$ is true.

In the following, we show that satisfying the specifications defined over the reduced-dimension state space implies the satisfaction of $\phi$. We assume that for all $s\in \mathcal{S}$, we have $\|p_{i}(t_0)-p_{j}(t_0)\|_{2} \geq \zeta$ for all $i,j\in \mathcal{N}_S$ with $i\neq j$, and $(p_{i}(t_0)-\overline{p}_{s}(t_0))^{T}\Sigma_{s}(t_0)^{-1}(p_{i}(t_0)-\overline{p}_{s}(t_0)) \leq 1$.
\begin{theorem}
    \label{theorem:centroid-planning-sufficient}
    Let $\phi = \psi \wedge \pi$ be defined as in \eqref{eq:STL}. Suppose that there exist $(\overline{x}_{1}(t),\Sigma_{1}(t)),\ldots,(\overline{x}_{S}(t),\Sigma_{S}(t))$ and $u_{1}(t),\ldots,u_{N}(t)$ that satisfy the following over the horizon $[t_0,t_0+T]$:
    \begin{enumerate}
        \item \label{theorem1:constraint1}$(\{(\overline{x}_s(t), \Sigma_s(t))\}_{s\in \mathcal{S}}, t_0) \models (\overline{\psi} \wedge G_{[t_{0},t_{0}+T]}\overline{\pi}_{1} \wedge G_{[t_{0},t_{0}+T]}\overline{\pi}_{2})$ 
        \item \label{theorem1:constraint2}For all agents $i$ and $j$ belonging to the same swarm, $\|p_{i}(t)-p_{j}(t)\|_{2} \geq \zeta$ for all $t\in [t_{0},t_{0}+T]$
        \item \label{theorem1:constraint3}For all agents $i$ belonging to swarm $s$, $$(p_{i}(t)-\overline{p}_{s}(t))^{T}\Sigma_{s}(t)^{-1}(p_{i}(t)-\overline{p}_{s}(t)) \leq 1$$ for all $t \in [t_{0},t_{0}+T]$.
    \end{enumerate}
    Then $(\mathbf{x},t) \models \phi$.
\end{theorem}
\begin{proof}
We divide the proof into two steps. First, we show that the conditions in Theorem $\ref{theorem:centroid-planning-sufficient}$ are sufficient to ensure that $({\mathbf{x}},t) \models \psi$ holds. Then, we prove that these conditions are also sufficient to guarantee $({\mathbf{x}},t) \models \pi$. Therefore, it follows that $({\mathbf{x}},t) \models \phi$.
    If $\{(\overline{x}_s(t), \Sigma_s(t))\}_{s\in \mathcal{S}} \models (\overline{\psi} \wedge G_{[t_{0},t_{0}+T]}\overline{\pi}_{1} \wedge G_{[t_{0},t_{0}+T]}\overline{\pi}_{2})$, we have that $\overline{\psi}$ is evaluated true and $\overline{\pi}_1$ and $\overline{\pi}_2$ are evaluated true for all $t\in[t_0,t_0+T]$.
    We first prove that if $\overline{\psi}$ is true, then $\psi$ is true. Suppose that $\overline{\mu}^{s}$ appears in $\overline{\psi}$ and is true. Define $x = [p^T,\ v^T]^T$ with $p\in \mathbb{R}^d$ and $v\in \mathbb{R}^d$. Since $\overline{h}^{\mu}_{s}(\overline{x}_{s},\Sigma_{s}) \geq 0$, we have that $h^{\mu}(x) \geq 0$ for all $x$ satisfying $(p-\overline{p}_{s})^{T}\Sigma_{s}^{-1}(p-\overline{p}_{s}) \leq 1$. By condition 3), this implies that $h^{\mu}(x_{i}) \geq 0$ for all $i \in \mathcal{N}_{s}$. Since $|\mathcal{N}_{s}| \geq N_{\mu}$, we have that atomic proposition $\mu$ is satisfied.

Next, suppose that $\neg \overline{\mu}^s$ is satisfied. By a similar argument, we have that at most one swarm satisfies $h^{\mu}(x_{i}) \geq 0$ for all $i \in \mathcal{N}_{s}$, that swarm has cardinality less than $N_{\mu}$, and all other swarms have $h^{\mu}(x_{i}) < 0$ for all $i \in \mathcal{N}_{s}$. Hence $|\{i : h^{\mu}(x_{i}) \geq 0\}| < N_{\mu}$, implying that $\neg \mu$ is satisfied. Since $\overline{\psi}$ evaluates to true with the satisfied predicates $\mu$ and $\neg \mu$, the formula $\psi$ will also evaluate to true by construction.
    
    We next prove that when $\overline{\pi}_1$ and $\overline{\pi}_2$ are evaluated true for all $t\in[t_0,t_0+T]$ and condition 2) holds, safety specification $\pi$ is also true.
    If $\overline{\pi}_1$ is true, we have that $\overline{h}_{ss'}^\mu\geq 0$.
    Therefore, for any agents $i,j\in\mathcal{N}$ belonging to different swarms $s\neq s'$, $\|p_i-p_j\|_2-\zeta\geq 0$ holds. 
    When condition 2) holds, by Definition \ref{def:safe-STL}, we have that $h_{ij}^\mu(x)\geq 0$ for all $t\in[t_0,t_0+T]$ for all agents $i$ and $j$ regardless whether they belong to the same swarm or not, and hence $\pi_1$ is true.
    If $\overline{\pi}_2$ is true, we have that $\overline{h}_{sj}^\theta\geq 0$.
    Consider an arbitrary agent $i$ within swarm $s$. 
    We have that $(p_i - \overline{p}_{s})^{T}\Sigma_{s}^{-1}(p_i-\overline{p}_{s}) \leq 1$.
    Since $\overline{h}_{sj}^\theta\geq 0$, then $h_{ij}^\theta\geq 0$ must also hold for all $i=1,\ldots,N$ and $j=1,\ldots,M$. 
    Consequently, $\pi_2$ is true when $\overline{\pi}_2$ holds. 
    Combining the arguments above, we have that $(\mathbf{x},t)\models\phi$ when the conditions in Theorem \ref{theorem:centroid-planning-sufficient} hold.
\end{proof}

Let $\{(\bar{x}_s(t),\Sigma_s(t))\}_{s\in \mathcal{S}}$ be a trajectory in the reduced-dimension state space.
We say the trajectory is \emph{valid} if the conditions in Theorem \ref{theorem:centroid-planning-sufficient} are met. 

In what follows, we use two steps to find a valid trajectory. In Step $1$, as described in Section \ref{sec:planning}, we first construct sufficient conditions for constraint \ref{theorem1:constraint1}) in Theorem \ref{theorem:centroid-planning-sufficient}. We then formulate the problem of finding $\{(\overline{x}_s(t), \Sigma_{s}(t))\}_{s\in \mathcal{S}}$ as an optimization problem. Finally, we propose a method to solve it. In Step $2$, as described in Section \ref{sec:low-level}, we propose a control policy for agents in swarms ensuring constraints \ref{theorem1:constraint2}) and \ref{theorem1:constraint3}) in Theorem \ref{theorem:centroid-planning-sufficient} are satisfied. Therefore, the trajectory synthesized in Step $1$ is valid.

\subsection{Planning in the Reduced-Dimension State Space}\label{sec:planning}
Using Theorem \ref{theorem:centroid-planning-sufficient}, we have that the original STL specification is satisfied if we can find a valid trajectory. 
In what follows, we abstract valid trajectories using a sequence of states in the reduced-dimension state space, denoted as a path.
We derive sufficient conditions under which trajectories traversing the path are valid.
We develop an algorithm for finding such paths. 
All proofs are deferred to Appendix.

Given the initial $\overline{x}_{s}(t_0)$ for all $s\in \mathcal{S}$ based on the given initial states $x_i(t_0)$ for all $i\in \{1, \ldots, N\}$. 
We define a path for a swarm $s\in\mathcal{S}$ in the reduced-dimension state space as follows. 
\begin{definition}\label{def:path}
    A finite sequence of time-stamped states, denoted as $\{(t_{s,k},\bar{x}_{s,k},\Sigma_{s,k})\}_{k=1}^{K_s}$, constitutes a path in the reduced-dimension state space if for all $k=1,\ldots,K_s$ there exist some control $u_i(t)$ for all agents $i$  
    for each swarm $s$ and time $t\in[t_{s,k-1},t_{s,k}]$ such that 
    \begin{enumerate}
    \item \label{def4con1}$t_{s,0} = t_0, \forall s\in \mathcal{S}$,
    \item \label{def4con2}$\overline{x}_{s}(t_{s,0}) = \overline{x}_{s}(t_0), \forall s\in\mathcal{S}$,
        \item \label{def4con3}$\bar{x}_s(t_{s,k-1})= \bar{x}_{s,k-1}$,
        \item \label{def4con4}$\bar{x}_s(t_{s,k})= \bar{x}_{s,k}$,
        \item \label{def4con5}$\bar{x}_s(t')=\bar{x}_{s,k-1}+\int_{t=t_{s,k-1}}^{t'}( A\overline{x}_{s}(t) + B\overline{u}_{s}(t))\dd t$ for all $t'\in[t_{s,k-1},t_{s,k}]$ and $k=1,\ldots,K_s$,
        \item \label{def4con6}$(p_i(t)-\overline{p}_s(t))^T\Sigma_{s,k}^{-1}(p_i(t)-\overline{p}_s(t))\leq 1$ holds for all agents $i$ in swarm $s$, $t\in[t_{s,k-1},t_{s,k}]$, and $k=1,\ldots,K_s$.
        \item \label{def4con7}$\Sigma_{s}(t) = \Sigma_{s,k}$, $\forall s\in \mathcal{S}$, $\forall t\in[t_{s,k-1},t_k]$
    \end{enumerate}
\end{definition}
Given Definition \ref{def:path}, the path of all swarms is expressed as $\mathcal{P}=\bigcup_{s\in\mathcal{S}}\bigcup_{k=1,\ldots,K_s}(t_{s,k},\bar{p}_{s,k},\Sigma_{s,k})$. In this paper, we define the transition process as the phase during which the bounding ellipsoid matrix changes from $\Sigma_{s,k-1}$ to $\Sigma_{s,k}$ for all $s\in \mathcal{S}$ and $k=1, \ldots, K_s$. Agents in swarm $s$ adjust their positions to relocate within the new bounding ellipsoid with the same centroid position $\bar{p}_{s,k}$. By the end of the transition process, all agents within a swarm stop at their new locations, i.e., $v_i(t_{s,k})=0$ for all $i\in \mathcal{N}_s$.
We assume that the transition process will not violate the conditions in Theorem \ref{theorem:centroid-planning-sufficient}. 

In what follows, we construct sufficient conditions to find a feasible $\mathcal{P}$. We propose a method to alternatively search over $(t_{s,k}, \bar{p}_{s,k})$ and $\Sigma_{s,k}$ for all $s\in \mathcal{S}$ and $k=1, \ldots, K_s$. Specifically, we first define $\epsilon$-valid trajectory by evaluating specifications in condition \ref{theorem1:constraint1}) in Theorem \ref{theorem:centroid-planning-sufficient} using under approximation functions. We evaluate $\bar{\pi}_1$, $\bar{\pi}_2$ and $\bar{\psi}$ by defining under approximation functions $\hat{h}_{ij}^{\mu}(\mathbf{x}(t)) = h_{ij}^{\mu}(\mathbf{x}(t))-\epsilon$ of $h_{ij}^{\mu}(\mathbf{x}(t))$, 
$\hat{h}_{ij}^{\theta}(\mathbf{x}(t)) = h_{ij}^{\theta}(\mathbf{x}(t))-\epsilon\|a_j\|_2$ of $h_{ij}^{\theta}(\mathbf{x}(t))$, $\hat{h}^{\mu}(x_i)=F^Tp_i+g-\epsilon\|F\|_2$ of $h^{\mu}(x_i)$, and $\hat{h}^{\neg\mu}(x_i)=-F^Tp_i-g-\epsilon\|F\|_2$ of $h^{\neg\mu}(x_i)$.
A nonnegative $\epsilon$ validates that under approximation conditions are sufficient conditions for the original ones. 

Next, we formulate the problem of generating a valid trajectory as follows:
\begin{align}
\label{problem:optimization-formulation}
\begin{array}{ll}
\underset{\mathcal{P}, \epsilon}{\max} & \epsilon \\
\mbox{s.t.} 
& \mbox{$\mathcal{P}$  is a path},\\
& \mbox{$(\overline{\mathbf{x}}(t), \Sigma(t))$ traversing $\mathcal{P}$ is $\epsilon$-valid.}
\end{array}
\end{align}

In the following, we derive a set of conditions to be satisfied by path $\mathcal{P}$ such that any trajectory $(\bar{\mathbf{x}}(t),\Sigma(t))$ traversing $\mathcal{P}$ satisfies specification $\phi$.
Specifically, we observe that given a path $\{(t_{s,k},\bar{x}_{s,k},\Sigma_{s,k})\}_{k=1}^{K_s}$, we can decompose any trajectory traversing the path into a sequence of segments, denoted as $\{\mathtt{SEG}_{s, k}\}_{k=1}^{K_s}$, where $\mathtt{SEG}_{s, k}$ represents the trajectory of $\bar{x}_s(t)$ for $t\in[t_{s,k-1},t_{s,k}]$ with a fixed bounding ellipsoid matrix $\Sigma_{s,k}$. 
We will derive sufficient conditions to characterize the behavior of each swarm at each segment to satisfy $\phi$. 

Considering the constraint that $\mathcal{P}$ is a path, the conditions \ref{def4con1}) to \ref{def4con4}) in Definition  \ref{def:path} are used to define the waypoints and corresponding timestamps for all segments. A time progression constraint will be automatically introduced. The condition \ref{def4con5}) is relaxed to a path reachability constraint. The condition \ref{def4con6}) is incorporated into the $\epsilon$-valid constraint with an intra-swarm collision avoidance condition. The condition \ref{def4con7}) enforces that the bounding ellipsoid matrices are preserved during the time interval $[t_{s,k-1},\ t_{s,k}]$ for all $s\in \mathcal{S}$ and $k=1, \ldots, K_s$.

Considering the constraint that $(\bar{\mathbf{x}}(t),\Sigma(t))$ traversing $\mathcal{P}$ is $\epsilon$-valid, we give sufficient conditions to the relaxed $\bar{\pi}_1$, relaxed $\bar{\pi}_2$, and relaxed $\bar{\psi}$, incorporating the condition \ref{def4con6}) in Definition \ref{def:path}.
\subsubsection{Time progression} To ensure the path $\mathcal{P}$ is well-defined, we require for all $s\in\mathcal{S}$ and for all $k = 0, \ldots, K_s-1$,  \begin{equation}\label{eq:time seq}
    t_0\leq t_{s,1}\leq \ldots\leq t_{s,k}\leq\ldots\leq t_{s,K_s}\leq t_0+T,
\end{equation} 
where $T$ denotes the pre-specified time limit for the task. 

\subsubsection{Path Reachability}
We first focus on path $\mathcal{P}$ and ensure trajectories traversing $\mathcal{P}$ are compatible with the dynamics of centroids in Eq. \eqref{eq:dynamics-centroid}.
Particularly, state $\bar{x}_{s,k}$ should be reachable from $\bar{x}_{s,k-1}$ within time $t_{s,k}-t_{s,k-1}$ for all $k=1,\ldots,K_s$ and swarm $s\in\mathcal{S}$. In the present paper, we implement a trajectory tracking controller for the centroid of swarms \cite{zheng2014trajectory,rodriguez2014trajectory}. We assume that the tracking controller ensures 1) the centroid is capable of reaching the next waypoint from the last one, and 2) the centroid remains within a distance $\eta$ of the line connecting two waypoints in a segment. Specifically, we define the states on the strict line connecting $x_{s,k-1}$ and $x_{s,k}$ as $l_{s,k}(\lambda)\triangleq (1-\lambda)x_{s,k-1}+\lambda x_{s,k}$, $\lambda\in [0,1]$. We assume the controller ensures $\min_{\lambda\in [0,1]}\|\overline{x}_s(t)-l_{s,k}(\lambda)\|_2\leq \eta$ for all $t\in [t_{s,k-1}, t_{s,k}]$. 
Similar to the setup in \cite{sun2022multi}, we require for all $s\in\mathcal{S}$ and for all $k = 0, \ldots, K_s-1$, 
\begin{equation}
\label{eq:velocity constraint}
    \|\overline{p}_{s,k+1}-\overline{p}_{s,k}\|_2\leq \chi \cdot (t_{s,k+1}-t_{s,k}),
\end{equation}
where $\chi$ denotes the limitation caused by actuation constraints.

\subsubsection{Satisfaction of relaxed $\bar{\pi}_1$}
In what follows, we consider the relaxed safety specification $\bar{\pi}_1$.
We need to ensure that $\overline{h}_{ss'}^{\mu}(\bar{x}_s, \bar{x}_{s^{\prime}})\geq \epsilon$ for all time $t\in[t_0,t_0+T]$ such that agents belonging to swarm $s\neq s'$ remain at least $(\zeta+\epsilon)$ distance away from each other.
Evaluating $\overline{h}_{ss'}^{\mu}(\bar{x}_s, \bar{x}_{s^{\prime}})\geq \epsilon$ involves computing $\|p_i-p_j\|_2$ for all pairs of agents $i$ and $j$ in swarms $s$ and $s'$, and all time $t \in [t_0, t_0+T]$.
To mitigate the computational challenge, we over-approximate the satisfaction of $\bar{\pi}_1$ by constructing a set, denoted as $\Gamma_{s,k}$, to bound the position of centroid $\bar{x}_s$ for swarm $s$ during $\mathtt{SEG}_{s,k}$. Based on the design of the trajectory tracking controller, we have $\Gamma_{s,k} = \{p\in \mathbb{R}^d: \min_{\lambda\in [0,1]}\|p-l_{s,k}(\lambda)\|_2\leq \eta\}$.

Our intuition is that if the sets $\Gamma_{s,k}$ and $\Gamma_{s',k'}$ of swarms $s$ and $s'$ are sufficiently far away from each other, then agents belonging to different swarms will also stay away from each other without creating any collision.

We next bound the distance between the sets $\Gamma_{s,k}$ and $\Gamma_{s',k'}$, and therefore distance between trajectories contained in these sets.
Let $s_1$ and $s_2$ be two distinct swarms.
A trivial condition to ensure that agents belonging to $s_1$ will not collide with $s_2$ is that 
\begin{equation}\label{eq:time nonoverlap}
    [t_{s_1, k_1-1}, t_{s_1,k_1}]\cap[t_{s_2, k_2-1}, t_{s_2,k_2}]=\emptyset.
\end{equation}
When $[t_{s_1, k_1-1}, t_{s_1,k_1}]\cap[t_{s_2, k_2-1}, t_{s_2,k_2}]\neq\emptyset$, we guarantee collision avoidance among agents from swarms $s_1$ and $s_2$ by deriving a lower bound for the distance between $\Gamma_{s,k}$ and $\Gamma_{s',k'}$.
We denote the distance between sets $\Gamma_{s_1,k_1}$ and $\Gamma_{s_2,k_2}$ as $\text{dist}(\Gamma_{s_1,k_1},\Gamma_{s_2,k_2})$. 
In the following, we derive a sufficient condition on $\text{dist}(\Gamma_{s_1,k_1},\Gamma_{s_2,k_2})$ such that $\bar{h}_{s_1s_2}^\mu\geq \epsilon$.

We define the bounding ellipsoid of swarm $s$ during $ \mathtt{SEG}_{s,k}$ as $\overline{\Omega}_{s,k}\triangleq\{\begin{bmatrix}
    p\\v
\end{bmatrix}\in\mathcal{X} :(p-\overline{p}_{s})^T\Sigma_{s,k}^{-1}(p-\overline{p}_{s})\leq 1\}$. Let $\lambda_{\text{max}}(\Sigma_{s,k})$ denote the maximum eigenvalue of matrix $\Sigma_{s,k}$.
\begin{lemma}\label{lemma:sufficient condition inter-swarm distance}
    If the following inequality holds for swarms $s_1\neq s_2$:
    \begin{multline}\label{eq:original inter-swarm constr}
    \|\frac{1}{2}(\bar{p}_{s_1,k_1} + \bar{p}_{s_1,k_1+1}) - \frac{1}{2}(\bar{p}_{s_2,k_2} - \bar{p}_{s_2,k_2+1})\|_1\geq \\ \|\frac{1}{2}(\bar{p}_{s_1,k_1} - \bar{p}_{s_1,k_1+1})\|_1 + \|\frac{1}{2}(\bar{p}_{s_2,k_2} - \bar{p}_{s_2,k_2+1})\|_1+\\
    (\sqrt{\lambda_{\text{max}}(\Sigma_{s_1,k_1})}+ \sqrt{\lambda_{\text{max}}(\Sigma_{s_2,k_2})}+ 2\eta + \zeta + \epsilon)\sqrt{d}
    ,
\end{multline}
then $\bar{h}_{s_1s_2}^\mu\geq \epsilon$.
\end{lemma}

The authors in \cite{sun2022multi} considered the size of agents and proposed a sufficient condition to ensure the collision avoidance for agents in different segments. Compared with their work, we consider the bounding ellipsoids information, instead of the size of a single agent, and introduce the feasible margin in order to find a feasible solution in \eqref{eq:original inter-swarm constr}.

Now the constraint to guarantee the satisfaction of $\bar{\pi}_1$ is formulated as  
\begin{align}
    \label{constraint:inter-swarm}
    \bigwedge _{\substack{s_{1}\neq s_{2}}}\bigwedge _{\substack{k_1=1,\ldots,K_{s_{1}}\\ k_2=1,\ldots,K_{s_{2}}}} \mathtt {safe}_{1} (\mathtt{SEG}_{s_1,k_1},  \mathtt{SEG}_{s_2,k_2}, \eta )
\end{align}
where
\begin{multline}\label{eq:safety 1 constr}
    \mathtt {safe}_{1} (\mathtt{SEG}_{s_1,k_2},  \mathtt{SEG}_{s_2,k_2}, \eta )\\
    = ([t_{s_1,k_1-1}, t_{s_1,k_1}] \cap [t_{s_2,k_2-1}, t_{s_2,k_2}] = \emptyset ) \\
    \vee
    \Big (\|\frac{1}{2}(\bar{p}_{s_1,k_1} + \bar{p}_{s_1,k_1+1}) - \frac{1}{2}(\bar{p}_{s_2,k_2} - \bar{p}_{s_2,k_2+1})\|_1\geq \\ \|\frac{1}{2}(\bar{p}_{s_1,k_1} - \bar{p}_{s_1,k_1+1})\|_1 + \|\frac{1}{2}(\bar{p}_{s_2,k_2} - \bar{p}_{s_2,k_2+1})\|_1\\
    +(\sqrt{\lambda_{\text{max}}(\Sigma_{s_1,k_1})}+ \sqrt{\lambda_{\text{max}}(\Sigma_{s_2,k_2})}+ 2\eta+\zeta+\epsilon)\sqrt{d}\Big). 
\end{multline}
The sufficiency of constraint \eqref{constraint:inter-swarm} to enforce $\bar{\pi}_1$ is shown below.
\begin{theorem}\label{thm: satisfaction of pi1}
If \eqref{constraint:inter-swarm}
    is true, then the safety specification $\bar{\pi}_1$ is guaranteed for all swarms. 
\end{theorem}
Theorem \ref{thm: satisfaction of pi1} follows from Lemma \ref{lemma:sufficient condition inter-swarm distance} and the time constraint. If two segments are disjoint in the time horizon, then the corresponding swarms avoid collision. Otherwise, based on the results in Lemma \ref{lemma:sufficient condition inter-swarm distance}, we have $\overline{h}_{s_1 s_2}^{\mu}\geq \epsilon$ hold, implying the collision avoidance between swarms. 
According to Theorem \ref{thm: satisfaction of pi1}, we have that if the time intervals associated with swarms $s_1$ and $s_2$ do not overlap (i.e., Eq. \eqref{eq:time nonoverlap} holds) or swarms $s_1$ and $s_2$ are sufficiently far from each other (i.e., \eqref{eq:safety 1 constr} holds) for all $s_1\neq s_2$, then safety specification $\pi_1$ is satisfied.

\subsubsection{Satisfaction of relaxed $\bar{\pi}_2$} We next focus on the relaxed safety specification $\bar{\pi}_2$.
Our goal is to derive constraints in the reduced-dimension state space such that $\bar{h}_{sj}^\theta(\bar{x}_s,\Sigma_s)\geq \epsilon\|a_j\|_2$ holds.
In what follows, we derive a sufficient condition to guarantee $\bar{h}_{sj}^\theta(\bar{x}_s,\Sigma_s)\geq \epsilon\|a_j\|_2$ and thereby $\gamma_j(x_i)=a_j^Tp_i+b_j-(\zeta+\epsilon)\|a_j\|_2\geq 0$ holds for all agents $i$ in swarm $s$.
\begin{lemma}\label{lemma:stl-swarm-inner}
    If $a_{j}^T\bar{p}_{s,k-1}+b_j - \sqrt{a_{j}^{T}\Sigma_{s,k} a_{j}} 
 - (\eta + \epsilon)\|a_j\|_2\geq 0$ and $a_{j}^T\bar{p}_{s,k}+b_j - \sqrt{a_{j}^{T}\Sigma_{s,k} a_{j}} 
 - (\eta +\epsilon)\|a_j\|_2\geq 0$ hold, then $\overline{h}_{s,j}^{\theta}(\bar{x}_s,\Sigma_s)\geq \epsilon\|a_j\|_2$ holds during $\mathtt{SEG}_{s,k}$.
\end{lemma}

Based on Lemma \ref{lemma:stl-swarm-inner}, we can thus guarantee $\bar{\pi}_2$ using the following constraint:
\begin{align}\label{constraint:swarm-safety}
    \bigwedge _{\substack{s\in\mathcal{S}}}\bigwedge _{\substack{k=1,\ldots,K_{s}}}\bigwedge _{\substack{j = 1, \ldots, M}} \mathtt {safe}_2 (\mathtt{SEG}_{s,k},  j, \epsilon ),
\end{align}
where
\begin{multline*}
    \mathtt {safe}_{2} (\mathtt{SEG}_{s,k},j,  \epsilon)= \\
    (a_{j}^Tp_{s,k}+b_j - \sqrt{a_{j}^{T}\Sigma_{s,k} a_{j}} -(\eta + \epsilon)\|a_j\|_2\geq 0 )\wedge\\
    (a_{j}^T\bar{p}_{s,k+1}+b_j - \sqrt{a_{j}^{T}\Sigma_{s,k} a_{j}} 
 - (\eta +\epsilon)\|a_j\|_2\geq 0). 
\end{multline*}

The sufficiency of constraint \eqref{constraint:swarm-safety} to satisfy $\bar{\pi}_2$ is shown as follows.
\begin{theorem}\label{thm: satisfaction of pi2}
    If constraint \eqref{constraint:swarm-safety} holds, then safety specification $\bar{\pi}_2$ is satisfied.
\end{theorem}
Theorem \ref{thm: satisfaction of pi2} follows from Lemma \ref{lemma:stl-swarm-inner}. The collision avoidance between swarms and obstacles is ensured by requiring each swarm to remain collision-free with all obstacles at every segment.

\subsubsection{Satisfaction of relaxed $\bar{\psi}$}\label{sec:swarm-stl}
In what follows, we derive sufficient conditions to ensure $\overline{\mu}^{s}$ holds for each swarm $s$ during $\mathtt{SEG}_{s,i}$ and hence swarm STL specification $\bar{\psi}$ is satisfied. 

Given a predicate function $h^\mu(x)=F^\mu p+g^\mu $, we denote the $r$-th row of $F^\mu$ and $g^\mu$ as $F_r^\mu$ and $g_r^\mu$, respectively.
We next derive sufficient conditions on $h^\mu(\bar{x}_s)$ for centroids of swarm $s$ during $\mathtt{SEG}_{s,k}$ such that $\bar{h}_s^\mu\geq 0$, thereby $\bar{\mu}^s$ is satisfied during the segment.
\begin{lemma}\label{lemma:stl-swarm-segments}
    For a given $\epsilon$  and the inequalities
    \begin{equation}
    F_{r}^\mu\bar{p}_{s,k}+g_{r}^\mu - \sqrt{F_{r}^\mu\Sigma_{s,k} F_{r}^{\mu^T}} - (\eta +\epsilon)\|F_r^{\mu}\|_2\geq 0
    \end{equation}
    \begin{equation}
    F_{r}^\mu\bar{p}_{s,k+1}+g_{r}^\mu - \sqrt{F_{r}^\mu\Sigma_{s,k} F_{r}^{\mu^T}} - (\eta +\epsilon)\|F_r^{\mu}\|_2\geq 0
    \end{equation}
    hold for all $r=1, \ldots, l$, then $\overline{\mu}^{s}$ is true for all $\overline{x}_{s}\in \Gamma_{s,k}$ during $\mathtt{SEG}_{s,k}$ with $1\leq k\leq K_{s}-1$. 
\end{lemma}

Considering the robustness, the relaxed negation predicates in $\overline{\psi}$ are defined as for a given $\epsilon$  and the inequalities
    \begin{equation}
    F_{r}^\mu\bar{p}_{s,k}+g_{r}^\mu + \sqrt{F_{r}^\mu\Sigma_{s,k} F_{r}^{\mu^T}} + (\eta +\epsilon)\|F_r^{\mu}\|_2\leq 0
    \end{equation}
    \begin{equation}
    F_{r}^\mu\bar{p}_{s,k+1}+g_{r}^\mu + \sqrt{F_{r}^\mu\Sigma_{s,k} F_{r}^{\mu^T}} +(\eta +\epsilon)\|F_r^{\mu}\|_2\leq 0
    \end{equation}
    hold for all $r=1, \ldots, l$, then $\neg\overline{\mu}^{s}$ is true for all $\overline{x}_{s}\in \Gamma_{s,k}$ during $\mathtt{SEG}_{s,k}$ with $1\leq k\leq K_{s}-1$. 

Next, we will formulate constraint $\mathtt{SwarmSTL}_{s,k}^{\overline{\psi}^{s}}$ associated with arbitrary swarm STL specification $\bar{\psi}$ by induction. 
The detailed induction is presented in Appendix \ref{sec:induct constraint}.
When the swarm STL specification involves disjunction between predicates, we use the big-M method and a set of binary variables \cite{griva2008linear}, denoted as $\mathbf{z}$, to convert it into conjunctive forms. 
For example, a swarm STL specification constraint in the disjunctive form $\lor_i (h_i\geq 0)$ is converted to $\land_i (h_i + (1-z_i)Y\geq 0) \land (\sum_i z_i\geq 1)$, where $z_i$ is binary and $Y$ is a sufficiently large number. 
Here, the binary variable $z_i=1$ indicates that  $h_i\geq 0$ holds.

\subsubsection{Intra-Swarm Collision Avoidance}
As shown in Theorem \ref{theorem:centroid-planning-sufficient}, agents within each swarm $s$ should stay at least $\zeta$ away from each other.
Therefore, each swarm $s$ should be sufficiently large to contain at least $\mathcal{N}_s$ agents in the bounding ellipsoid.
We approximate this requirement by representing each swarm using a ball with radius $\zeta$, and constrain the volume of the bounding ellipsoid to be $\xi>1$ times larger than the total volume of all agents.
This constraint is formalized as follows:
\begin{align}
\label{constraint:bounding-box}
    \log(\det(\Sigma_{s,k}))\geq 2\cdot log(\xi |\mathcal{N}_{s}|\zeta^{d})
    ,~\forall s\in\mathcal{S},~k=1,\ldots,K_s.
\end{align}

\subsubsection{Optimization for Planning in Reduced-Dimension State Space}\label{sec:algo}

Combining our discussion in the previous sections, we can reformulate the optimization problem \eqref{problem:optimization-formulation} as follows:

    \begin{align}
\label{eq:optimization-formulation}
\underset{\mathcal{P}, \epsilon}{\max} & \quad \epsilon\\
\mbox{s.t.} 
&\quad \text{Actuation-induced constraint: Eq. \eqref{eq:velocity constraint}}\nonumber\\
&\quad \text{Time progression constraint: Eq. \eqref{eq:time seq}}\nonumber\\
&\quad \text{Satisfaction of $\bar{\pi}_1$: Eq. \eqref{constraint:inter-swarm}}\nonumber\\
&\quad \text{Satisfaction of $\bar{\pi}_2$: Eq. \eqref{constraint:swarm-safety}}\nonumber\\
 &\quad \text{Intra-swarm collision avoidance: Eq. \eqref{constraint:bounding-box}}\nonumber\\
& \quad \text{Swarm STL specification }\mathtt{SwarmSTL}_{s,k}^{\overline{\psi}^{s}} \mbox{ holds}\nonumber
\end{align}

Some constraints in this optimization problem have a convex form instead of a linear form, e.g., in constraint \eqref{constraint:inter-swarm}, making it a mixed-integer convex programming instead of MILP when we use the Big-M method. This makes it computationally challenging to solve.
We address this challenge using an alternating algorithm shown in Algorithm \ref{algo:waypoint}.
Instead of solving for all decision variables in one pass, the algorithm iteratively uses two steps to optimize a subset of decision variables with others being fixed.

At each iteration, the algorithm first solves for the centroid of each swarm $\bar{x}_{s,k}$, time stamps $t_{s,k}$, binary variables $\mathbf{z}$, and $\epsilon$ with fixed bounding ellipsoid $\Sigma_{s,k}$.
This requires solving the mixed-integer linear program \eqref{eq:waypoint optimization} in Appendix \ref{app:optimization}.
Next, we update bounding ellipsoids $\Sigma_{s,k}$ and $\epsilon$ for all $s$ and $k$ with fixed centroids $\bar{x}_{s,k}$, time stamps $t_{s,k}$, and binary variables $\mathbf{z}$.
The bounding ellipsoids can be found by solving a semi-definite program as shown in Eq. \eqref{eq:box optimization}.
The iteration terminates when $\epsilon\geq 0$ holds after the optimization over bounding ellipsoids.
We characterize Algorithm \ref{algo:waypoint} as follows.

\begin{theorem}\label{thm:convergence}
    Algorithm \ref{algo:waypoint} converges within finitely many iterations.
    If Algorithm \ref{algo:waypoint}  terminates with $\epsilon \geq 0$, then the returned path $\mathcal{P}$ is valid and the STL specification is satisfied.
\end{theorem}

\begin{center}
  	\begin{algorithm}[!htp]
  		\caption{Alternating algorithm to solve Eq. \eqref{eq:optimization-formulation}}
  		\label{algo:waypoint}
  		\begin{algorithmic}[1]
            \State \textbf{Initialize:} The bounding ellipsoid $\Sigma_{s,k}^0$ for all $s$ and $k$.
            \State Set iteration index $\tau\leftarrow 1$, choose the maximum number of iterations $\tau_{max}$
  			\While {$\epsilon<0$}
            \State Update $\epsilon$, binary variables $\mathbf{z}^\tau$, centroids $\bar{p}_{s,k}^\tau$, and time stamps $t_{s,k}^\tau$ for all $s$ and $k$ by solving Eq. \eqref{eq:waypoint optimization} with fixed $\Sigma_{s,k}^{\tau-1}$ for all $s$ and $k$
            \State Update $\epsilon$ and bounding ellipsoids $\Sigma_{s,k}^\tau$ by solving Eq. \eqref{eq:box optimization} with fixed $\mathbf{z}^\tau$, $\bar{p}_{s,k}^{\tau}$, and $t_{s,k}^\tau$ for all $s$ and $k$
            \If {$\tau\geq \tau_{max}$}
            \State \textbf{break}
            \EndIf
            \EndWhile
            \State \textbf{Return} $t_{s,k}^\tau$, $\bar{p}_{s,k}^\tau$, and $\Sigma_{s,k}^\tau$
  		\end{algorithmic}
  	\end{algorithm}
\end{center}

\subsection{Low-level Control}\label{sec:low-level}

\input{Low_level}

%% file: Low_level.tex
In this section, we derive the control input $u_i$ of each agent $i$ such that agents form swarms which traverse a valid path $\mathcal{P}$ we designed in Section \ref{sec:planning}, and ensure the conditions \ref{theorem1:constraint2} and  \ref{theorem1:constraint3} in Theorem \ref{theorem:centroid-planning-sufficient} are met. Based on the assumption in Section \ref{sec:planning}, these conditions are satisfied during the transition process, and $[\overline{p}_s(t_{s,k})^T,\ \overline{v}_s(t_{s,k})^T]^T=\overline{x}_{s,k}$. In this section, we derive the control input $u_i$ for all $i\in \mathcal{N}_s$, that ensures all agents in swarm $s$ avoid collisions and remain within the bounding ellipsoid throughout the interval $t\in[t_{s,k-1},\ t_{s,k}]$ for all $s\in\mathcal{S}$ and all $k=1, \ldots, K_s$.

The control input of a centroid $\overline{u}_s(t)$ is determined by the trajectory tracking control policy introduced in Section \ref{sec:planning}. In what follows, we will prove the conditions \ref{theorem1:constraint2}) and \ref{theorem1:constraint3}) in Theorem \ref{theorem:centroid-planning-sufficient} are met under the low-level control $u_i(t) = \overline{u}_s(t)$ for all $t\in [t_{s,k-1},\ t_{s,k}]$.

\begin{theorem}
    \label{theorem:low-level-agent-control}
For swarm $s\in \mathcal{S}$, if for all $i\in \mathcal{N}_s$ ,$u_i(t)=\overline{u}_{s}(t)$ holds for all $t\in [t_{s,k-1},\ t_{s,k}]$,  $\forall k = 1, \ldots, K_s$, then we have that 
$\|p_{i}(t)-p_{j}(t)\|_{2} \geq \zeta$ for all $i,j\in \mathcal{N}_s$ with $i\neq j$, and $$(p_{i}(t)-\overline{p}_{s}(t))^{T}\Sigma_{s}(t)^{-1}(p_{i}(t)-\overline{p}_{s}(t)) \leq 1$$ for all $t\in [t_{0},t_{0}+T]$.
\end{theorem}
\begin{proof}
    Based on the assumption in Section \ref{sec:planning}, the transition process will not violate the conditions in Theorem \ref{theorem:centroid-planning-sufficient} and the initial condition for each time interval $t\in [t_{s,k-1}, \ t_{s,k}]$ ensures $\|p_i(t_{s,k-1})-p_j(t_{s,k-1})\|\geq \zeta$, $\forall i,j\in \mathcal{N}_s$ with $i\neq j$, and $(p_{i}(t_{s,k-1})-\overline{p}_{s}(t_{s,k-1}))^{T}\Sigma_{s,k}^{-1}(p_{i}(t_{s,k-1})-\overline{p}_{s}(t_{s,k-1})) \leq 1$, for all $i\in \mathcal{N}_s$.

    Next, we consider the satisfaction of conditions during time intervals $t\in [t_{s,k-1}, \ t_{s,k}]$. Define $\Psi_{i,j}(t) = (p_i(t)-p_j(t))^T(p_i(t)-p_j(t))$, we have $\Psi_{i,j}(t_{s,k-1})\geq \zeta^2$ and $\dot{\Psi}_{ij}(t) = 2(v_i(t)-v_j(t))^T(p_i(t)-p_j(t))$. We have that $\dot{x}_i-\dot{x}_j=A(x_j-x_j)$, implying $\dot{v}_i-\dot{v}_j=0$. Given $v_i(t_{s,k-1}) = v_j(s_{s,k-1})=0$, we have $v_i(t)=v_j(t)=\overline{v}_s(t)$ for all $t\in [t_{s,k-1}, \ t_{s,k}]$. Therefore, we have $\dot{\Psi}_{ij}(t)=0$ for all $t\in [t_{s,k-1}, \ t_{s,k}]$, implying $\Psi_{i,j}(t_{s,k-1})\geq \zeta^2$ for all $t\in [t_{s,k-1}, \ t_{s,k}]$.

    Define  $\overline{\Psi}_{i,k}(t)=(p_{i}(t)-\overline{p}_{s}(t))^{T}\Sigma_{s,k}(t)^{-1}(p_{i}(t)-\overline{p}_{s}(t))-1$. In what follows, we prove $\overline{\Psi}_{i,k}(t)\leq 0$ for all $t\in [t_{s,k-1}, \ t_{s,k}]$. We have $\overline{\Psi}_{i,k}(t_{s,k-1})\leq 0$ and $\dot{\overline{\Psi}}_{i,k}(t)=(v_{i}(t)-\overline{v}_{s}(t))^{T}\Sigma_{s,k}(t)^{-1}(p_{i}(t)-\overline{p}_{s}(t))+(p_{i}(t)-\overline{p}_{s}(t))^{T}\Sigma_{s,k}(t)^{-1}(v_{i}(t)-\overline{v}_{s}(t))=0$, implying $\overline{\Psi}_{i,k}(t)\leq 0$ for all $t\in [t_{s,k-1}, \ t_{s,k}]$.

    Therefore, if we have $u_i(t)=\overline{u}_s(t)$ for all $i\in\mathcal{N}_S$ for all $t\in [t_{s,k-1}, \ t_{s,k}]$, $k=1, \ldots, N_s$, we have $\|p_{i}(t)-p_{j}(t)\|_{2} \geq \zeta$ for all $i,j\in \mathcal{N}_s$ with $i\neq j$, and $$(p_{i}(t)-\overline{p}_{s}(t))^{T}\Sigma_{s}(t)^{-1}(p_{i}(t)-\overline{p}_{s}(t)) \leq 1$$ for all $t\in [t_{0},t_{0}+T]$, completing the proof.
\end{proof}

Theorem \ref{theorem:low-level-agent-control} shows that if there exists some control $\bar{u}_s$ such that the centroid of each swarm can transition from $\bar{x}_{s,k-1}$ to $\bar{x}_{s,k}$ and satisfy the safety specification $\pi_2$, then there exist a control input for each agent such that the agents remain within the swarm without colliding with each other.

Therefore, if the waypoints, timestamps, and bounding ellipsoids are chosen by the approach in Section \ref{sec:planning}, and the low-level control policy aforementioned is followed, then the generated $(\textbf{x}, t)$ satisfies all the constraints in Theorem \ref{theorem:centroid-planning-sufficient}, and solves Problem \ref{prob:formulation}.

%% file: Simulation.tex
\section{Simulation}
\label{sec:simulation}
In this section, we evaluate our proposed method on two benchmarks. 
We compare our proposed method with the baseline \cite{sun2022multi} in both single and multiple swarm scenarios. We use Python to implement the algorithm, with the Gurobi optimizer \cite{gurobi} used for solving optimization problems when optimizing over waypoints and timestamps, while the Mosek optimizer \cite{mosek} is employed to solve the optimization problem when optimizing over bounding ellipsoids due to the convenience of its atomic functions, such as the $\text{log}(\text{det}(X))$ function.
All experiments were conducted on MacOS with Apple M3 Pro chip and 32 GB RAM.

The goal of this section includes two folds. We first demonstrate our proposed method by exhibiting that the STL tasks are achieved and the safety requirements are satisfied. Secondly, we evaluate the additional computational overhead introduced by incorporating bounding ellipsoids. We demonstrate that the computational overhead in our approach scales efficiently, whereas the baseline approach becomes intractable as the number of agents increases. 
 
\subsection{Swarm path planning}
\label{sec:benchmarks}
\input{tables/trajectories}
We consider both single-swarm and multi-swarm scenarios in this paper. 
In the single-swarm scenario, we consider the benchmark stlcg-1 in \cite{leung2023backpropagation} and \cite{sun2022multi}. As shown in Fig. \ref{fig:stlcg-1}, the swarm starts from the bottom-left corner and moves toward the top-right corner. During the process, all the agents in the swarm need to visit the blue ($\mathcal{B}$) regions and green ($\mathcal{G}$) regions and avoid the red ($\mathcal{R}$). 
The STL for this task can be specified as swarm STL $(F_{[0,T]}G_{[0,5]}\mathcal{B})\wedge(F_{[0,T]}G_{[0,5]}\mathcal{G})\wedge(G_{[0,T]}\neg \mathcal{R})$ for all agents with pre-specified time limit $T$. We evaluate the predicate based on Definition \ref{def:swarm predicate}. In this scenario, we set the number of agents in the swarm as $5$, which satisfies the requirements for all tasks. The path planning of the centroid is shown in Fig. \ref{fig:stlcg-1}. The results demonstrate successful task execution, and the collision is avoided.

The multi-swarm scenario is designed based on the benchmark wall-1 from \cite{sun2022multi}. Multiple swarms are tasked with reaching goal regions while avoiding collisions with both walls and the other swarms, especially in the vicinity of the narrow door. We use $\mathcal{W}$ to denote the wall and $\mathcal{G}_1, \ldots, \mathcal{G}_4$ to denote the goals. The corresponding goal regions are marked with green squares from the left top to the right top. The corresponding requirements for the number of agents are $20$, $5$, $25$, and $10$, respectively. The four swarms contain $20$, $5$, $15$, and $30$ agents, respectively. The task can be specified as a swarm STL as 
$\wedge_{j=1}^{4}\vee_{i=1}^4 (G_{[0,T]}\neg \mathcal{W}) \wedge (F_{[0,T]}\mathcal{G}_j)$. The path planning of the centroids for four swarms is shown in Fig. \ref{fig:wall-1-auto}. The centroids reach the goal regions without collision. In this scenario, the swarms choose their goal regions automatically based on the number of agents they have and the requirements for the tasks.
\input{tables/path-figs}
\subsection{Comparison with the baseline}
We compare our proposed method with the baseline approach in \cite{sun2022multi} for both single-swarm and multi-swarm scenarios.
 Considering the limited space of the environment and the large number of agents, we set $\eta = 0.05$ and $\zeta=0.01$ for both our proposed approach and the baseline in both scenarios to satisfy the spacing constraints in both scenarios. 

In the single-swarm scenario, we compare the computational time required for single-swarm planning as described in Section \ref{sec:benchmarks} with that of single-agent planning using the baseline approach. In the baseline approach, a single agent is considered to satisfy the task specification $(F_{[0,T]}G_{[0,5]}\mathcal{B})\wedge(F_{[0,T]}G_{[0,5]}\mathcal{G})\wedge(G_{[0,T]}\neg \mathcal{R})$ within the stlcg-1 environment, which is described in the single-swarm case in Section \ref{sec:benchmarks}. The predicate in the baseline is evaluated based on the classic definition of STL in \cite{donze2010robust}. This comparison highlights the additional time introduced by considering the shapes of the bounding ellipsoids in both waypoints and timestamps optimization and bounding ellipsoids optimization, demonstrating that the algorithm maintains computational efficiency even as the number of agents increases. 
In this scenario, we set the number of segments as $9$, which is consistent with the baseline.  
The total computational time for our approach is $0.52$ seconds, including $0.28$ seconds for waypoints and timestamps optimization and $0.24$ seconds for bounding ellipsoids optimization. The baseline uses $0.06$ seconds for the waypoints and timestamps planning for one agent. Compared with the baseline, our approach uses an extra $0.22$ seconds to plan for waypoints and timestamps for a swarm with $5$ agents.
This demonstrates that our method maintains computational tractability even as the number of agents increases, making it suitable for larger-scale applications.

Next, we compare the computational time of our proposed method with that of the baseline approach in the wall-1 environment, which is described in the multi-swarm case in Section \ref{sec:benchmarks}. 
In the baseline approach, agents are assigned to specific goal regions. To ensure a fair comparison, we assign the goal regions to specific swarms when using our proposed method.
The baseline approach initially computes the waypoints and timestamps for four agents. 
We gradually increase the number of agents by sequentially positioning new agents to the left of each initialized agent and share the same goal region with it. A fixed spacing of $0.2$ meters is maintained between them. 
In our approach, agents assigned to the same goal region are treated as a swarm, and their centroid positions are initialized accordingly. The bounding ellipsoids are initialized with $\Sigma_{s,k}=\begin{bmatrix}
    0.01 &0\\0 & 0.01
\end{bmatrix}$ for all $s\in \mathcal{S}$ and $k=1, \ldots, K_s$. Both the baseline approach and our proposed method use the same number of segments, which is set to $6$. 

\input{tables/baseline-comparsion-new}
As shown in Table \ref{table: computational_time}, when the number of agents increases, the computational time varies within the range of $35.6$ to $44.7$ seconds. The variation remains limited, as we consistently plan for a fixed number of four swarms. In the baseline approach, when the number of agents is fewer than $6$, the computational time is lower than that of our proposed method. However, when the number of agents reaches $6$, the computational time of the baseline approach increases significantly, which is approximately $3.7$ times that of our proposed method. In the cases with $7$ or more agents, the baseline approach fails to compute a solution within $900$ seconds.
These results demonstrate that, compared with the baseline approach, our method maintains a tractable computational time for a fixed number of swarms, even as the number of agents increases.

We observe that the computational time of the waypoints and timestamps optimization in our approach is sensitive to the initialization of the bounding ellipsoids.
Additionally, we note that, in some of the simulation results, the baseline approach results in a lower computation time compared to our proposed method. This is due to a difference in the optimization problems solved by the two methods. In the baseline approach, the feasible margin $\epsilon$ is fixed, while the goal is to minimize the total completion time. In contrast, our approach has the feasible margin $\epsilon$ as an optimization variable. As a comparison, when we fix the feasible margin of our approach at $\epsilon = 0.001$ and minimize the completion time, we obtain a computation time of $0.42$s for four swarms. 

%% file: tables/trajectories.tex
\begin{figure}
    \centering
    \includegraphics[width = 0.25\textwidth]{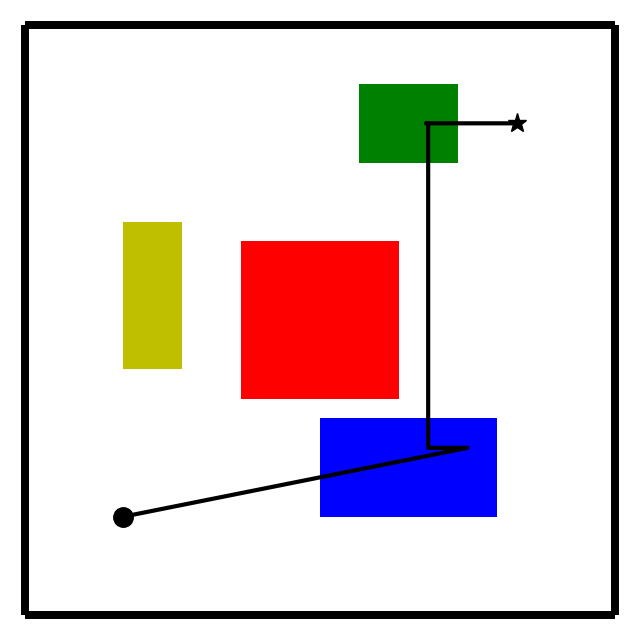}
    \caption{The path planning for the single-swarm scenario. The black lines represent the planned trajectory for the centroid of the swarm. The starting position of the centroid is indicated by a circle, while a star is used to mark the end. The blue and green regions denote goal areas that must be visited, while the red region indicates an unsafe area that all agents must avoid at all times. }
    \label{fig:stlcg-1}
\end{figure}

%% file: tables/path-figs.tex
\begin{figure}
    \centering
    \includegraphics[width = 0.4\textwidth]{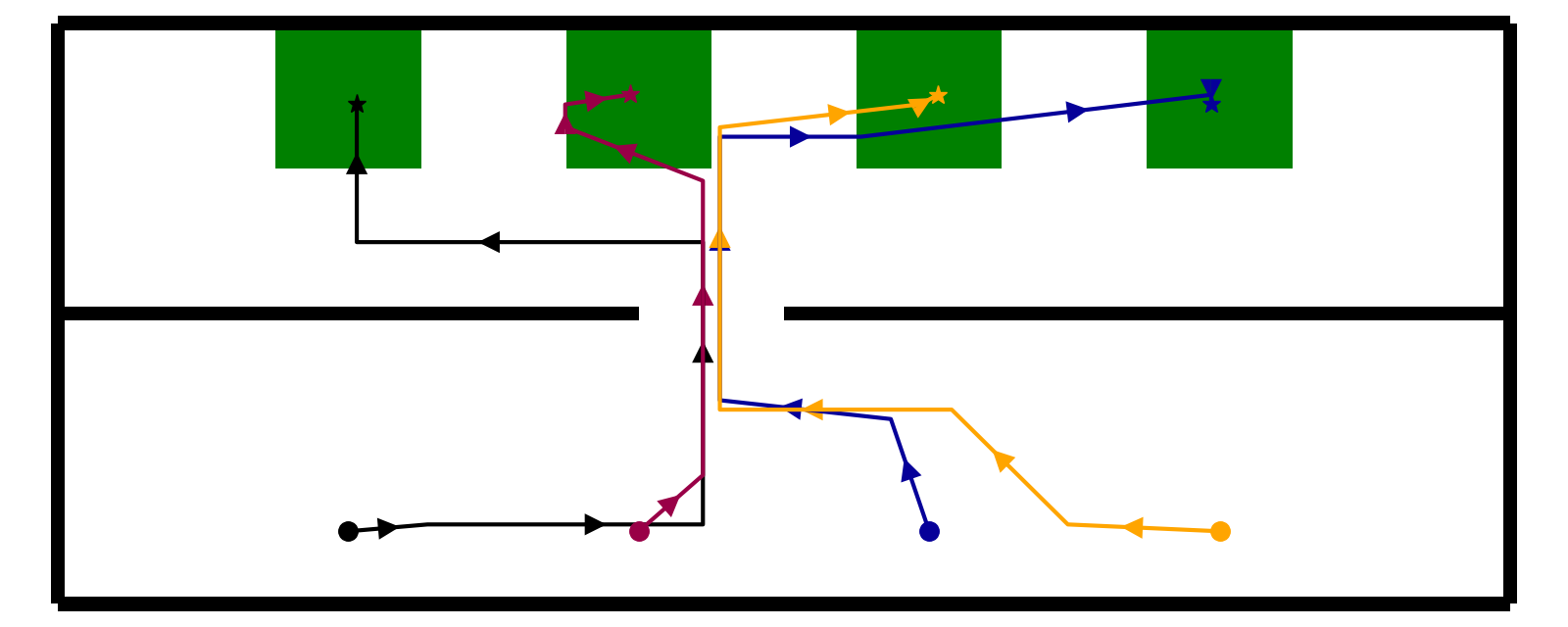}
    \caption{The path planning for the multi-swarm scenario. Goal regions are automatically assigned to swarms based on the task-specific agent requirements and the number of agents within each swarm. The swarms move toward their goal regions without collision. Solid lines denote the path planning for the centroids of swarms. The starting point of the centroid is marked with a circle on the trajectory, and a star is used to mark the end. Arrowheads are added to indicate the direction of motion.}
    \label{fig:wall-1-auto}
\end{figure}

%% file: tables/baseline-comparsion-new.tex
\begin{table}[!htp]
\begin{tabular}{|c|ccc|c|}
\hline
\multirow{2}{*}{Agent \#} & \multicolumn{3}{c|}{Our approach (s)}                                                                                                                                                      & \multirow{2}{*}{Baseline (s)} \\ \cline{2-4}
                          & \multicolumn{1}{c|}{\begin{tabular}[c]{@{}c@{}}Waypoints and\\  timestamps\end{tabular}} & \multicolumn{1}{c|}{\begin{tabular}[c]{@{}c@{}}Bounding \\ ellipsiod\end{tabular}} & Total time &                               \\ \hline
4                         & \multicolumn{1}{c|}{39.3}                                                                & \multicolumn{1}{c|}{3.8}                                                           & 43.1       & 0.6                           \\ \hline
5                         & \multicolumn{1}{c|}{40.9}                                                                & \multicolumn{1}{c|}{3.8}                                                           & 44.7       & 26.2                          \\ \hline
6                         & \multicolumn{1}{c|}{34.0}                                                                & \multicolumn{1}{c|}{3.7}                                                           & 37.8       & 140.9                         \\ \hline
7                         & \multicolumn{1}{c|}{37.0}                                                                & \multicolumn{1}{c|}{3.7}                                                           & 40.6       & \textgreater{}900             \\ \hline
8                         & \multicolumn{1}{c|}{31.1}                                                                & \multicolumn{1}{c|}{4.5}                                                           & 35.6       & \textgreater{}900             \\ \hline
\end{tabular}
\caption{Comparison of computational time between our approach and the baseline. Swarms in our approach and agents in the baseline are assigned specific goal regions, respectively. Our
approach is significantly faster than
the baseline when the number of agents is greater than or equal to $6$. 
}
\label{table: computational_time}
\end{table}

%% file: Conclusion.tex
\section{Conclusion}\label{sec: conclusion}
In this paper, we investigated a large-scale multi-agent system subject to an STL specification. To address the scalability challenge, we decomposed the problem into two steps. First, we abstracted agents in a swarm to a centroid with a bounding ellipsoid. The swarm STL was designed to ensure task completion and safety at the swarm level. Next, we proposed a low-level control policy for individual agents within each swarm to guarantee reachability and collision avoidance. In the experiments, we considered both single-swarm and multi-swarm scenarios. The achievement of tasks and the guarantee of safety demonstrated our algorithm. The computational time remained scalable with an increasing number of agents, provided that the number of swarms was fixed.

%% file: appendix.tex
\appendix

This section presents the construction of $\mathtt{SwarmSTL}_{s,k}^{\overline{\psi}^{s}}$, the proofs, and the algorithm we proposed in Section \ref{sec:planning}.
\subsection{Construction of $\mathtt{SwarmSTL}_{s,k}^{\overline{\psi}^{s}}$}\label{sec:induct constraint}
The construction of $\mathtt{SwarmSTL}_{s,k}^{\overline{\psi}^{s}}$ depends on the expression of $\overline{\psi}$.
Therefore, we discuss how to construct $\overline{\psi}$ based on the following scenarios.
Let $\bar{\mu}_1$ and $\bar{\mu}_2$ be swarm STL specifications.
\begin{itemize}
    \item If $\overline{\psi}^s=\bar{\mu_1}^s\land\bar{\mu_2}^s$, then $\mathtt{SwarmSTL}_{s,k}^{\overline{\psi}^{s}}=\mathtt{SwarmSTL}_{s,k}^{\overline{\mu}_1^s}\land\mathtt{SwarmSTL}_{s,k}^{\overline{\mu}_2^s}$.
    \item If $\overline{\psi}^s=\bar{\mu_1}^s\lor\bar{\mu_2}^s$, then $\mathtt{SwarmSTL}_{s,k}^{\overline{\psi}^{s}}=\mathtt{SwarmSTL}_{s,k}^{\overline{\mu}_1^s}\lor\mathtt{SwarmSTL}_{s,k}^{\overline{\mu}_2^s}$.
    \item If $\overline{\psi}^s=G_{[a,b]}\bar{\mu}^s$, then 
    \begin{multline*}
          \mathtt{SwarmSTL}_{s,k}^{\overline{\psi}^{s}}
          =\bigwedge _{k'=1}^{K_s}\Big([t_{s,k'-1}, t_{s,k'}] \\
          \cap [t_{s,k-1}+a, t_{s,k}+b] \neq \emptyset \implies \mathtt{SwarmSTL}_{s,k}^{\overline{\mu}^s} \Big).
    \end{multline*}
    \item If $\overline{\psi}^s=F_{[a,b]}\bar{\mu}^s$, then 
    \begin{multline*}
          \mathtt{SwarmSTL}_{s,k}^{\overline{\psi}^{s}}
          =(t_{k} - t_{k-1} \leq b-a) \wedge \\ 
           \bigvee _{k'=1}^{K_s-1} \left([t_{k'-1}, t_{k'}] \cap [t_{k}+a, t_{k-1}+b] \neq \emptyset \right.\\
          \left.\wedge \mathtt{SwarmSTL}_{s,k}^{\overline{\mu}^s} \right).
    \end{multline*}
    \item If $\overline{\psi}^s=\bar{\mu}_1^sU_{[a,b]}\bar{\mu}_2^s$, then 
    \begin{multline*}
          \mathtt{SwarmSTL}_{s,k}^{\overline{\psi}^{s}}
          =(t_{k} - t_{k-1} \leq b-a) \wedge \\ \bigvee _{k'=1}^{K_{s}-1} \Bigg([t_{k'-1}, t_{k'}] \cap [t_{k}+a, t_{k-1}+b] \neq \emptyset \wedge \mathtt{SwarmSTL}_{s,k}^{\overline{\mu}_2^s} \\ \wedge \bigwedge _{k''=1}^{k'} \left([t_{k''-1}, t_{k''}] \cap [t_{k-1}, t_{k}+b] \neq \emptyset \right.\\
          \left.\implies \mathtt{SwarmSTL}_{s,k}^{\overline{\mu}_1^s}\right) \Bigg).
    \end{multline*}
\end{itemize}

\subsection{Proof of Lemma \ref{lemma:sufficient condition inter-swarm distance}}
Based on the results in Section III.B in \cite{sun2022multi}, the condition in \eqref{eq:original inter-swarm constr} can be intuitively interpreted as a sufficient condition to ensure the distance between two centers of two segments is greater than the sum of half the lengths of the two segments plus a margin. Considering the bounding boxes of two swarms, the margin is set as  $(\sqrt{\lambda_{\text{max}}(\Sigma_{s_1,k_1})}+ \sqrt{\lambda_{\text{max}}
(\Sigma_{s_2,k_2})}+ 2\eta + \zeta+\epsilon)$. 
Next, we prove by contradiction. If there exists two points $x_1$ and $x_2$ on two segments $\mathtt{SEG}_{s_1,k_1}$ and $\mathtt{SEG}_{s_2,k_2}$, respectively, such that $\|x_1-x_2\|_2<\overline{\epsilon}\triangleq\sqrt{\lambda_{\text{max}}(\Sigma_{s_1,k_1})}+ \sqrt{\lambda_{\text{max}}
(\Sigma_{s_2,k_2})}+ 2\eta + \zeta+\epsilon$, then we have 
$\frac{1}{2}(\bar{x}_{s_1,k_1} + \bar{x}_{s_1,k_1+1}) - \frac{1}{2}(\bar{x}_{s_2,k_2} - \bar{x}_{s_2,k_2+1})= (\frac{1}{2}(\bar{x}_{s_1,k_1} + \bar{x}_{s_1,k_1+1})-x_1)+(x_1-x_2)+ (x_2 - \frac{1}{2}(\bar{x}_{s_2,k_2} - \bar{x}_{s_2,k_2+1}))$, implying 
\begin{multline*}
    \| \frac{(\bar{x}_{s_1,k_1} + \bar{x}_{s_1,k_1+1})}{2} - \frac{(\bar{x}_{s_2,k_2} - \bar{x}_{s_2,k_2+1})}{2}\|_1\leq \|x_1-x_2\|_1\\
    + \|\frac{(\bar{x}_{s_1,k_1} + \bar{x}_{s_1,k_1+1})}{2}-x_1\|_1 + \|x_2 - \frac{(\bar{x}_{s_2,k_2} - \bar{x}_{s_2,k_2+1})}{2}\|_1\\
    <\sqrt{d}\overline{\epsilon}+\|\frac{1}{2}(\bar{x}_{s_1,k_1} - \bar{x}_{s_1,k_1+1})\|_1 + \|\frac{1}{2}(\bar{x}_{s_2,k_2} - \bar{x}_{s_2,k_2+1})\|_1,
\end{multline*}
which leads to a contradiction. We also have that the distance between the centroid of swarm $s_1$ and $\mathtt{SEG}_{s_1,k_1}$ is less than or equal to $\eta$, and the same holds for $s_2$ and $\mathtt{SEG}_{s_2,k_2}$. The distances between agents with the centroid in swarm $s_1$ and $s_2$ are less than or equal to $\sqrt{\lambda_{\text{max}}(\Sigma_{s_1,k_1})}$ and $ \sqrt{\lambda_{\text{max}}
(\Sigma_{s_2,k_2})}$, respectively. Therefore, the condition in \eqref{eq:original inter-swarm constr} implies that the distance between any two agents form $s_1$ and $s_2$ in $\mathtt{SEG}_{s_2,k_2}$ and $\mathtt{SEG}_{s_2,k_2}$, respectively, is greater than or equal to $\epsilon$, i.e., $\bar{h}_{s_1s_2}^\mu\geq \epsilon$, completing the proof. 
   
\subsection{Proof of Theorem \ref{thm: satisfaction of pi1}}
If $\bigwedge _{\substack{s_{1}\neq s_{2}}}\bigwedge _{\substack{k_1=1,\ldots,K_{s_{1}}\\ k_2=1,\ldots,K_{s_{2}}}} \mathtt {safe}_1 (\mathtt{SEG}_{s_1,k_1},  \mathtt{SEG}_{s_2,k_2}, \allowbreak \epsilon)$ is true, we have $\mathtt {safe}_1 (\mathtt{SEG}_{s_1,k_1},  \mathtt{SEG}_{s_2,k_2}, \allowbreak \epsilon)$ is true for all $(s_1,s_2,k_1,k_2)$ with $s_1\neq s_2$, implying $[t_{s_1,k_1-1}, t_{s_1,k_1}] \cap [t_{s_2,k_2-1}, t_{s_2,k_2}] = \emptyset $ or $\|\frac{1}{2}(\bar{p}_{s_1,k_1} + \bar{p}_{s_1,k_1+1}) - \frac{1}{2}(\bar{p}_{s_2,k_2} - \bar{p}_{s_2,k_2+1})\|_1\geq  \|\frac{1}{2}(\bar{p}_{s_1,k_1} - \bar{p}_{s_1,k_1+1})\|_1 + \|\frac{1}{2}(\bar{p}_{s_2,k_2} - \bar{p}_{s_2,k_2+1})\|_1+ (\sqrt{\lambda_{\text{max}}(\Sigma_{s_1,k_1})}+ \sqrt{\lambda_{\text{max}}(\Sigma_{s_2,k_2})}+ 2\eta + \zeta +\epsilon)\sqrt{d}.$
If $\mathtt{SEG}_{s_1,k_1}$ and $\mathtt{SEG}_{s_2,k_2}$ are disjoint in the time horizon, then swarm $s_1$ and $s_2$ avoid collision, implying the guarantee of $\overline{\pi}_1$.
Otherwise, based on the results in Lemma \ref{lemma:sufficient condition inter-swarm distance}, we have $\overline{h}_{s_1 s_2}^{\mu}\geq \epsilon$ hold.

\subsection{Proof of Lemma \ref{lemma:stl-swarm-inner}}
Based on the trajectory tracking control, for any $\overline{p}_s\in \Gamma_{s,k}$, there exists $\bar{p} = \gamma \bar{p}_{s,k-1} + (1-\gamma)\bar{p}_{s,k}$ with $\gamma\in [0,1]$, such that $\|\overline{p}_s - \overline{p}\|_2\leq \eta$.
Given $a_{j}^T\bar{p}_{s,k}+b_j - \sqrt{a_{j}^{T}\Sigma_{s,k} a_{j}} 
 - (\eta + \epsilon)\|a_j\|_2\geq 0$ and $a_{j}^T\bar{p}_{s,k}+b_j - \sqrt{a_{j}^{T}\Sigma_{s,k} a_{j}} 
 - (\eta +\epsilon)\|a_j\|_2\geq 0$, we have $a_{j}^T\bar{p}+b_j - \sqrt{a_{j}^{T}\Sigma_{s,k} a_{j}} 
 - (\eta +\epsilon)\|a_j\|_2\geq 0$. This implies $a_{j}^T\bar{p}_s+b_j - \sqrt{a_{j}^{T}\Sigma_{s,k} a_{j}} 
 - \epsilon\|a_j\|_2\geq 0$. For any agent $i$ in the swarm, $a_{j}^Tp_i+b_{j} \geq a_j^T\overline{p}_{s,k}+b_j-\sqrt{a_j^T\Sigma_{s,k}a_j}\geq \epsilon\|a_j\|_2$ holds for all $p_i\in \{p\in\mathbb{R}^d: (p-\overline{p}_{s,k})\Sigma_{s,k}^{-1}(p-\overline{p}_{s,k})\leq 1\}$. This implies $\overline{h}_{s,j}^{\theta}(\bar{x}_s,\Sigma_s)\geq \epsilon\|a_j\|_2$ holds during $\mathtt{SEG}_{s,k}$.

 \input{tables/optimization_problem}   
\subsection{Proof of Theorem \ref{thm: satisfaction of pi2}}

If \eqref{constraint:swarm-safety} holds, we have $\mathtt {safe}_{2} (\mathtt{SEG}_{s,k},j,  \epsilon)$ holds for all $s$, $k$, and $j$. Based on the results in Lemma \ref{lemma:stl-swarm-inner}, we have $\overline{h}_{s,j}^{\theta}(\bar{x}_s,\Sigma_s)\geq \epsilon\|a_j\|_2$ hold during $\mathtt{SEG}_{s,k}$. 
Next, we consider the satisfaction of $\overline{\pi}_{2}$ during transition at the starting point of $\mathtt{SEG}_{s,k}$. 
Given that $\overline{\pi}_2$ is guaranteed during the process of bounding ellipsoid transitions,
then we have $\overline{h}_{s,j}^{\theta}(\bar{x}_s,\Sigma_s)\geq \epsilon\|a_j\|_2$ is guaranteed.

\subsection{Proof of Lemma \ref{lemma:stl-swarm-segments}}

    The proof of Lemma \ref{lemma:stl-swarm-segments} is similar to that of Lemma \ref{lemma:stl-swarm-inner}, and hence is omitted due to space constraint.

\subsection{Proof of Theorem \ref{thm:convergence}}

The termination guarantee of Algorithm \ref{algo:waypoint} follows when $\tau_{max}$ is finite.
The validity of path $\mathcal{P}$ follows from Definition \ref{def:path}.
We next prove the satisfaction of STL specification.
When a feasible solution is returned by Algorithm \ref{algo:waypoint} with $\epsilon\geq 0$, we have that all constraints of Eq. \eqref{eq:optimization-formulation} are met.
When $\epsilon\geq 0$, we have that constraints \eqref{constraint:inter-swarm} and \eqref{constraint:swarm-safety} hold, implying that $\mathtt {safe}_{1} (\mathtt{SEG}_{s_1,k_2},  \mathtt{SEG}_{s_2,k_2}, \eta )$ and $\mathtt {safe}_{2} (\mathtt{SEG}_{s_1,k_2},  j, \epsilon )$ are satisfied.
Using Theorem \ref{thm: satisfaction of pi1} and \ref{thm: satisfaction of pi2}, we thus have that safety specification $\bar{\pi}_1$ and $\bar{\pi}_2$ are satisfied. 
Furthermore, when $\mathtt{SwarmSTL}_{s,k}^{\overline{\psi}^{s}}$ is satisfied, then the swarm STL specification $\bar{\psi}$ is fulfilled.
Given our STL specification $\phi=\psi\land\pi$ and Theorem \ref{theorem:centroid-planning-sufficient}, we thus have that the STL specification is satisfied.

\subsection{Algorithm Description}\label{app:optimization}
We introduce the slack variable $\epsilon$ as described in Section \ref{sec:planning} when solving optimization problem \eqref{eq:optimization-formulation}.
This yields the following optimization problem:
    \begin{align} \underset{\substack{\epsilon,t_{s,k},\bar{x}_{s,k},\Sigma_{s,k},\\s\in\mathcal{S},k=1,\ldots,K_s}}{\max} &\quad \epsilon \label{eq: relax app optimization}\\
\mbox{s.t.} 
&\quad \text{Actuation-induced constraint: Eq. \eqref{eq:velocity constraint}}\nonumber\\
&\quad \text{Time progressing constraint: Eq. \eqref{eq:time seq}}\nonumber\\
&\quad \text{Satisfaction of $\bar{\pi}_1$: Eq. \eqref{constraint:inter-swarm}}\nonumber\\
&\quad \text{Satisfaction of $\bar{\pi}_2$: Eq. \eqref{constraint:swarm-safety}}\nonumber\\
 &\quad \text{Intra-swarm collision avoidance: Eq. \eqref{constraint:bounding-box}}\nonumber\\
& \quad \text{Swarm STL specification }\mathtt{SwarmSTL}_{s,k}^{\overline{\psi}^{s}} \mbox{ holds}\nonumber
\end{align}

We summarize the optimization problem used to update $\epsilon$, $\mathbf{z}$, $\bar{x}_{s,k}^\tau$, and $t_{s,k}^{\tau}$ for all $s$ and $k$ in Eq. \eqref{eq:waypoint optimization}. We further relax the constraint \eqref{eq:velocity constraint} as 
    $\|\overline{x}_{s,k+1}-\overline{x}_{s,k}\|_1\leq \chi \cdot (t_{s,k+1}-t_{s,k}),~\forall s\in \mathcal{S}, k=0,1,\ldots, K_s-1$.
Note that the relaxed constraint is a sufficient condition to constraint \eqref{eq:velocity constraint} and can be transfer to a linear form of $\bar{x}_{s,k}$ and $\bar{x}_{s,k+1}$.

Therefore, the optimization problem \eqref{eq:waypoint optimization} is a mixed integer linear program.

We present the optimization problem used to update bounding ellipsoid $\Sigma_{s,k}$ for all $s$ and $k$ in Eq. \eqref{eq:box optimization}.
To solve this optimization problem, we rewrite $\Sigma_{s,k}$ as $\sigma_{s,k}\sigma_{s,k}^T$ and convert the decision variable as $\sigma_{s,k}$. 
This allows us to represent $\sqrt{\lambda_{\max}(\Sigma_{s,k})}$ as $\|\sigma_{s,k}\|_2$ and represent $\sqrt{a_{j}^{T}\Sigma_{s,k} a_{j}}$ as $\|a_j^T\sigma_{s,k}\|_2$.
We can thus observe that the optimization problem is a semi-definite program, and thus can be solved efficiently.

%% file: tables/optimization_problem.tex
\begin{table*}[t]
\begin{subequations}\label{eq:waypoint optimization}
    \begin{align}
\underset{\substack{\epsilon,\mathbf{z},\bar{p}_{s,k}, t_{s,k},\\s\in\mathcal{S},k=1,\ldots,K_s}}{\max} &\quad \epsilon\\
\mbox{s.t.} 
& \quad \|\overline{p}_{s,k+1}-\overline{p}_{s,k}\|_1\leq \chi \cdot (t_{s,k+1}-t_{s,k}),~\forall s\in \mathcal{S}, ~\forall k\in \{0, 1, \ldots, K_s-1\},\\
&\quad t_0\leq t_{s,1}\leq \ldots\leq t_{s,k}\leq\ldots\leq t_{s,K_s}\leq t_0+T,~\forall s\in\mathcal{S}\\
&\quad ([t_{s_1,k_1-1}^{\tau}, t_{s_1,k_1}^{\tau}] \cap [t_{s_2,k_2-1}^{\tau}, t_{s_2,k_2}^{\tau}] = \emptyset ) 
    \vee
    \Big (\|\frac{\bar{p}_{s_1,k_1} + \bar{p}_{s_1,k_1+1}}{2} - \frac{\bar{p}_{s_2,k_2} - \bar{p}_{s_2,k_2+1}}{2}\|_1\geq \|\frac{\bar{p}_{s_1,k_1} - \bar{p}_{s_1,k_1+1}}{2}\|_1 \nonumber\\
    &+ \|\frac{\bar{p}_{s_2,k_2} - \bar{p}_{s_2,k_2+1}}{2}\|_1 + 
    (\sqrt{\lambda_{\text{max}}(\Sigma_{s_1,k_1})}+ \sqrt{\lambda_{\text{max}}(\Sigma_{s_2,k_2})}+ 2\eta+\zeta+\epsilon)\sqrt{d}\Big), ~\forall s_1\neq s_2, ~k_1=1,\ldots,K_{s_1},~k_2=1,\ldots,K_{s_2}\label{eq:relaxed constr 1}\\
&\quad a_{j}^T\bar{p}_{s,k}+b_j - \sqrt{a_{j}^{T}\Sigma_{s,k}^\tau a_{j}} - 
    (\eta + \epsilon)\|a_j\|_2\geq 0 ,~\forall s\in\mathcal{S},~k=1,\ldots,K_s,~j=1,\ldots,M\\
& \quad\mathtt{SwarmSTL}_{s,k}^{\overline{\psi}^{s}} \mbox{\ is satisfied},~\forall s\in\mathcal{S},k=1,\ldots,K_s
\end{align}
\end{subequations}
\begin{subequations}\label{eq:box optimization}
    \begin{align}
\underset{\substack{\epsilon,\Sigma_{s,k},\\s\in\mathcal{S},k=1,\ldots,K_s}}{\max} &\quad \epsilon\\
\mbox{s.t.} 
    &\quad ([t_{s_1,k_1-1}^{\tau}, t_{s_1,k_1}^{\tau}] \cap [t_{s_2,k_2-1}^{\tau}, t_{s_2,k_2}^{\tau}] = \emptyset ) 
    \vee
    \Big (\|\frac{\bar{p}_{s_1,k_1} + \bar{p}_{s_1,k_1+1}}{2} - \frac{\bar{p}_{s_2,k_2} - \bar{p}_{s_2,k_2+1}}{2}\|_1\geq \|\frac{\bar{p}_{s_1,k_1} - \bar{p}_{s_1,k_1+1}}{2}\|_1 \nonumber\\
    &\quad+ \|\frac{\bar{p}_{s_2,k_2} - \bar{p}_{s_2,k_2+1}}{2}\|_1 + 
    (\sqrt{\lambda_{\text{max}}(\Sigma_{s_1,k_1})}+ \sqrt{\lambda_{\text{max}}(\Sigma_{s_2,k_2})}+ 2\eta+\zeta + \epsilon)\sqrt{d}\Big), ~\forall s_1\neq s_2, ~k_1=1,\ldots,K_{s_1},~k_2=1,\ldots,K_{s_2}\\
    &\quad a_{j}^T\bar{p}_{s,k}+b_j - \sqrt{a_{j}^{T}\Sigma_{s,k}^\tau a_{j}} - 
    (\eta + \epsilon)\|a_j\|_2\geq 0 ,~\forall s\in\mathcal{S},~k=1,\ldots,K_s,~j=1,\ldots,M\\
    &\quad \log(\det(\Sigma_{s,k}))\geq 2\cdot \log(\eta |\mathcal{N}_{s}|b^{n})
    ,~\forall s\in\mathcal{S},~k=1,\ldots,K_s,\\
& \quad\mathtt{SwarmSTL}_{s,k}^{\overline{\psi}^{s}} \mbox{\ is satisfied},~\forall s\in\mathcal{S},k=1,\ldots,K_s
\end{align}
\end{subequations}
\end{table*}